\documentclass[a4paper,11pt]{article}
\pdfoutput=1

\usepackage{geometry}
\usepackage[margin=10mm]{caption}

\usepackage[utf8]{inputenc}
\usepackage[english]{babel}
\usepackage{amssymb, amsmath}
\usepackage{amsthm}
\usepackage{eucal}
\usepackage{mathalpha}
\usepackage{amsfonts}
\usepackage{mathrsfs}
\usepackage{setspace}
\usepackage{graphicx}
\usepackage{subcaption}
\usepackage{slashed}
\usepackage[sorting=none, backend=biber]{biblatex}
\usepackage{array,booktabs}
\usepackage{lmodern}
\usepackage{tikz}
\usepackage{microtype}
\usepackage{hyperref}
\usepackage{authblk}
\usepackage{yfonts}
\usepackage{csquotes}

\usepackage{tikz-cd}
\usetikzlibrary{shapes.geometric}
\usepackage{caption}
\usetikzlibrary{calc}
\def\centerarc[#1](#2)(#3:#4:#5)
    { \draw[#1] ($(#2)+({#5*cos(#3)},{#5*sin(#3)})$) arc (#3:#4:#5); }

\theoremstyle{definition}

\newtheorem{remark}{Remark}[section]

\newtheorem{prop}{Proposition}[section]
\newtheorem{lemma}{Lemma}[section]

\makeatletter
\renewenvironment{proof}[1][\proofname]{%
   \par\pushQED{\qed}\normalfont%
   \topsep6\p@\@plus6\p@\relax
   \trivlist\item[\hskip\labelsep\bfseries#1\@addpunct{.}]%
   \ignorespaces
}{%
   \popQED\endtrivlist\@endpefalse
}

\newenvironment{hproof}{%
  \renewcommand{\proofname}{Proof sketch}\proof}{\endproof}

\newcommand{\beq}{\begin{equation}}
\newcommand{\eeq}{\end{equation}}
\newcommand{\eeeem}{\end{multline}}
\newcommand{\bem}{\begin{multline}}
\newcommand{\bqa} {\begin{eqnarray}}
\newcommand{\eqa} {\end{eqnarray}}
\newcommand{\bmul}{\begin{multline}}
\newcommand{\emul}{\end{multline}}

\newcommand{\p}{\partial}

\newcommand{\eps}{\varepsilon}

\renewcommand\Re{\operatorname{Re}}
\renewcommand\Im{\operatorname{Im}}

\newcommand{\CA}{{\mathcal A}}
\newcommand{\CB}{{\mathcal B}}

\newcommand{\CH}{{\mathcal H}}

\newcommand{\CK}{{\mathcal K}}

\newcommand{\CO}{{\mathcal O}}

\newcommand{\CU}{{\mathcal U}}
\newcommand{\CV}{{\mathcal V}}
\newcommand{\CW}{{\mathcal W}}

\newcommand{\NN}{{\mathbb N}}
\newcommand{\ZZ}{{\mathbb Z}}
\newcommand{\RR}{{\mathbb R}}
\newcommand{\CCC}{{\mathbb C}}

\newcommand{\SA}{{\mathscr A}}
\newcommand{\SAl}{{\mathscr A}_{\ell}}

\def \l {\left(}
\def \r {\right)}

\def \lal {\langle}
\def \ral {\rangle}

\newcommand{\ObRepV}{{I}}
\newcommand{\voa}{{\mathsf V}}

\newcommand{\Hleg}{{\hat{\CH}}}

\newcommand{\holeO}{{\CO}}

\newcommand{\siteloc}{{z}}

\newcommand{\dLat}{{\Lambda^{\vee}}}
\newcommand{\clusterweight}{{\CW}}
\newcommand{\BoundedOps}{{\CB}}

\newcommand{\PsiNB}{{\Psi}}
\newcommand{\emb}{{\varphi}}
\newcommand{\anyonemb}{{\upsilon}}
\newcommand{\eqberry}{{{\boldsymbol\sigma}}}

\title{Chiral topologically ordered states on a lattice from vertex operator algebras}

\author{Nikita Sopenko \smallskip\\
{\it California Institute of Technology, Pasadena, CA 91125, USA}}

\date{\today}

\begin{document}

\maketitle

\begin{abstract}
We propose a class of pure states of two-dimensional lattice systems realizing topological order associated with unitary rational vertex operator algebras. We show that the states are well-defined in the thermodynamic limit and have exponential decay of correlations. The construction provides a natural way to insert anyons and compute certain topological invariants. It also gives candidates for bosonic states in non-trivial invertible phases, including the $E_8$ phase.
\end{abstract}

\section{Introduction}   \label{sec:intro}

It is believed that topological phases of matter in two dimensions can be classified by $(2+1)$-dimensional unitary topological quantum field theories (TQFT) and that the whole information about this TQFT is encoded in the entanglement structure of the ground state. In particular, for any such topological field theory, there should be a pure state of a lattice model on $\RR^2$ that has the corresponding topological order. While a precise mathematical characterization of states that are supposed to be related to TQFT and the procedure that allows to extract the TQFT data have not been given, there are certain classes of states, such as the class of invertible states introduced by A. Kitaev \cite{Kitaevlecture} or a class of gapped states\footnote{It is often assumed that topologically ordered states must be gapped ground states of local Hamiltonians. This condition is not sufficient to ensure the connection with TQFT as the classification of $d$-dimensional gapped states drastically diverges from 
that of topological field theories, at least starting from $d=3$. Neither does it seem necessary, as there might be a weaker assumption that already ensures topological properties. In this note, we concentrate on the topological properties of states that can be described without parent Hamiltonians.}, for which various topological invariants associated with TQFT have been defined.

There is a big class of non-chiral topological orders that can be realized by tensor network states defined using the Turaev-Viro construction. The string-net models introduced by M. Levin and X.-G. Wen \cite{LevinWen,LevinWenTensorNetwork,LevinWenTensorNetworVidal} provide parent commuting projector Hamiltonians for such states. The situation is very different for chiral topological order as only a few chiral interacting topological models have been solved exactly \cite{kitaev2006anyons}. One of the main features of such states is the non-existence of a boundary with short-range correlations. When the boundary modes of the system with such a ground state can be described by a conformal field theory, this obstruction manifests itself in the central charge of the chiral Virasoro algebra.

It was suggested a long time ago that the wave function of electrons for the ground state of a fractional quantum hall system is related to some vertex operator algebra \cite{moore1991nonabelions}, generalizing Laughlin's wave function. Similar ideas have been used to propose various candidates for chiral topologically ordered states of lattice spin systems \cite{kalmeyer1987equivalence, schroeter2007spin, ciracLaughlin1}. While some expected properties of such states can be checked numerically, it is difficult to verify analytically that they are representatives of the correct topological phase. To the best of our knowledge, no general proposal exists for an arbitrary unitary TQFT. Even a lattice state for the simplest non-trivial invertible bosonic phase, known as Kitaev's $E_8$ phase, has not been described explicitly.

In this note, we propose class of states of lattice systems with infinite-dimensional on-site Hilbert spaces\footnote{The entanglement of each site in the construction is finite. The states can be well approximated by replacing the infinite-dimensional spaces with finite-dimensional ones of large enough dimension.} associated with any good unitary rational vertex operator algebra. We show that the states from this class are well-defined in the thermodynamic limit and have exponentially decaying correlations. We argue that they have the topological order associated with the vertex operator algebra and propose a way to insert anyons. In particular, we obtain a lattice realization of a fractional quantum Hall state with such a topological order\footnote{For an integer quantum Hall phase, a commuting projector Hamiltonian model has been proposed in \cite{demarco2021commuting}}. We argue that our states belong to the class for which flux insertions and the topological invariant associated with the Hall conductance can be defined rigorously \cite{bachmann2019many,ThoulessHall}, and show that the latter coincides with the expected value. The usage of infinite-dimensional on-site Hilbert spaces allows us to use the full power of conformal invariance and get analytic results by relating state averages of observables of a lattice system to correlators of a certain auxiliary tensor network statistical model that was recently introduced in \cite{Runkel21}. It also allows us to map the problems of computing invariants of states and determining the properties of topological excitations to evaluations of correlation functions in CFT which can be done exactly. We emphasize that our chiral states are not tensor network states. However, we show that they can be obtained as a limit of a family of non-chiral tensor network states together with a decoupled complex conjugated copy.

The paper is organized as follows. In Section \ref{sec:schematic}, we give a schematic description of the construction. In Section \ref{sec:general}, after introducing some notation and terminology, we provide a more detailed treatment and argue that the states are topologically ordered. In particular, we discuss the case of a holomorphic vertex operator algebra that gives examples of invertible states of lattice systems. The specialization to the algebra of free fermions and free bosons is discussed in Section \ref{sec:examples}. Technical details are presented in the appendices.
\\

\noindent
{\bf Acknowledgements:}
I would like to thank Alexei Kitaev for discussions and Michael Levin for comments on the draft. I am especially grateful to Anton Kapustin for many useful suggestions and constant support. This research was supported by the U.S.\ Department of Energy, Office of Science, Office of High Energy Physics, under Award Number DE-SC0011632 and Dominic Orr Graduate Fellowship at Caltech.
\\

\noindent
{\bf Data availability statement:}
Data sharing not applicable to this article as no datasets were generated or analysed during the current study.
\\

\section{Schematic description} \label{sec:schematic}

\begin{figure}
    \centering
    \begin{tikzpicture} [x=13pt,y=13pt]
        \def\a{.4};
        \def\b{.6};
        \def\z{.25}
        \filldraw[color=teal!30] (-5-\a*5,-5) -- (-5+\a*5,5) -- (5+\a*5,5)--(5-\a*5,-5);
        \foreach \x in {-5,-3,...,5}
            \foreach \y in {-5,-3,...,5}
                {
                \filldraw [color=white, fill=white]  (\x-\b+\a*\y-\a*\b,\y-\b) [rounded corners = 6 pt] -- (\x-\b+\a*\y+\a*\b,\y+\b) [sharp corners] -- (\x+\b+\a*\y+\a*\b,\y+\b) [rounded corners = 6 pt] -- (\x+\b+\a*\y-\a*\b,\y-\b) [sharp corners] -- (\x-\b+\a*\y-\a*\b,\y-\b);
                }
        \foreach \x in {-5,-3,...,5}
            \foreach \y in {-5,-3,...,5}
                {
                \draw [color=teal!70]  (\x-\b+\a*\y-\a*\b,\y-\b) [rounded corners = 6 pt] -- (\x-\b+\a*\y+\a*\b,\y+\b) [sharp corners] -- (\x+\b+\a*\y+\a*\b,\y+\b);
                \draw [color=teal!70](\x+\b+\a*\y+\a*\b+\a*\z,\y+\b+\z) [rounded corners = 6 pt] -- (\x+\b+\a*\y-\a*\b,\y-\b) [sharp corners] -- (\x-\b+\a*\y-\a*\b-\z,\y-\b);     
                }
        \filldraw [color=white, fill=white] (-6-6*\a,-6)--(-6-5*\a,-5)--(6-5*\a,-5)--(6-6*\a,-6);
        \filldraw [color=white, fill=white] (-6+6*\a,+6)--(-6+5*\a,+5)--(6+5*\a,+5)--(6+6*\a,+6);
        \filldraw [color=white, fill=white] (6-6*\a,-6)--(6+6*\a,6)--(5+6*\a,6)--(5-6*\a,-6);    
        \filldraw [color=white, fill=white] (-6-6*\a,-6)--(-6+6*\a,6)--(-5+6*\a,6)--(-5-6*\a,-6);            
        \foreach \x in {-4,-2,...,4}
            {
            \filldraw[color = teal!30, fill=teal!30] (\x+5*\a,+5) ellipse (0.40 and 0.2);
            }        
        \foreach \x in {-4,-2,...,4}
            {
            \filldraw[fill=teal!30] (\x-5*\a,-5) ellipse (0.40 and 0.2);
            }
        \foreach \y in {-4,-2,...,4}
            {
            \filldraw[color=teal!30,rotate around={65:(-5+\y*\a,\y)}] (-5+\y*\a,\y) + (0:0.40 and 0.2) arc   (0:360:0.40 and 0.2);
            }   
        \foreach \y in {-4,-2,...,4}
            {
            \filldraw[color=black, fill = teal!30,rotate around={65:(5+\y*\a,\y)}] (5+\y*\a,\y) + (0:0.40 and 0.2) arc   (0:360:0.40 and 0.2);
            }               
        \draw[color=violet,densely dashed] (0-1*\a,-1) + (0:0.40 and 0.2) arc   (0:180:0.40 and 0.2); 
        \draw[color=violet,densely dashed] (0+1*\a,+1) + (0:0.40 and 0.2) arc   (0:180:0.40 and 0.2);
        \draw[color=violet,rotate around={65:(-1,0)},densely dashed] (-1,0) + (0:0.40 and 0.2) arc   (0:180:0.40 and 0.2);
        \draw[color=violet,rotate around={65:(1,0)},densely dashed] (1,0) + (0:0.40 and 0.2) arc   (0:180:0.40 and 0.2);   
    \end{tikzpicture}
    \begin{tikzpicture} [x=13pt,y=13pt]
        \def\a{.4};
        \def\b{.6};
        \def\z{.25}
        \filldraw[color=teal!30] (-5-\a*5,-5) -- (-5+\a*5,5) -- (5+\a*5,5)--(5-\a*5,-5);
        \foreach \x in {-5,-3,...,5}
            \foreach \y in {-5,-3,...,5}
                {
                \filldraw [color=white, fill=white]  (\x-\b+\a*\y-\a*\b,\y-\b) [rounded corners = 6 pt] -- (\x-\b+\a*\y+\a*\b,\y+\b) [sharp corners] -- (\x+\b+\a*\y+\a*\b,\y+\b) [rounded corners = 6 pt] -- (\x+\b+\a*\y-\a*\b,\y-\b) [sharp corners] -- (\x-\b+\a*\y-\a*\b,\y-\b);
                }
        \foreach \x in {-5,-3,...,5}
            \foreach \y in {-5,-3,...,5}
                {
                \draw [color=teal!70]  (\x-\b+\a*\y-\a*\b,\y-\b) [rounded corners = 6 pt] -- (\x-\b+\a*\y+\a*\b,\y+\b) [sharp corners] -- (\x+\b+\a*\y+\a*\b,\y+\b);
                \draw [color=teal!70](\x+\b+\a*\y+\a*\b+\a*\z,\y+\b+\z) [rounded corners = 6 pt] -- (\x+\b+\a*\y-\a*\b,\y-\b) [sharp corners] -- (\x-\b+\a*\y-\a*\b-\z,\y-\b);     
                }
        \filldraw [color=white, fill=white] (-6-6*\a,-6)--(-6-5*\a,-5)--(6-5*\a,-5)--(6-6*\a,-6);
        \filldraw [color=white, fill=white] (-6+6*\a,+6)--(-6+5*\a,+5)--(6+5*\a,+5)--(6+6*\a,+6);
        \filldraw [color=white, fill=white] (6-6*\a,-6)--(6+6*\a,6)--(5+6*\a,6)--(5-6*\a,-6);    
        \filldraw [color=white, fill=white] (-6-6*\a,-6)--(-6+6*\a,6)--(-5+6*\a,6)--(-5-6*\a,-6);            
        \foreach \x in {-4,-2,...,4}
            {
            \filldraw[color = teal!30, fill=teal!30] (\x+5*\a,+5) ellipse (0.40 and 0.2);
            }        
        \foreach \x in {-4,-2,...,4}
            {
            \filldraw[fill=teal!30] (\x-5*\a,-5) ellipse (0.40 and 0.2);
            }
        \foreach \y in {-4,-2,...,4}
            {
            \filldraw[color=teal!30,rotate around={65:(-5+\y*\a,\y)}] (-5+\y*\a,\y) + (0:0.40 and 0.2) arc   (0:360:0.40 and 0.2);
            }   
        \foreach \y in {-4,-2,...,4}
            {
            \filldraw[color=black, fill = teal!30,rotate around={65:(5+\y*\a,\y)}] (5+\y*\a,\y) + (0:0.40 and 0.2) arc   (0:360:0.40 and 0.2);
            }               
        \draw[color=violet,densely dashed] (0-1*\a,-1) + (0:0.40 and 0.2) arc   (0:180:0.40 and 0.2); 
        \draw[color=violet,densely dashed] (0+1*\a,+1) + (0:0.40 and 0.2) arc   (0:180:0.40 and 0.2);
        \draw[color=violet,rotate around={65:(-1,0)},densely dashed] (-1,0) + (0:0.40 and 0.2) arc   (0:180:0.40 and 0.2);
        \draw[color=violet,rotate around={65:(1,0)},densely dashed] (1,0) + (0:0.40 and 0.2) arc   (0:180:0.40 and 0.2);

        \foreach \x\y in {-3/-3,3/3}
                {
                \def\sc{0.1}
                \draw [color=red, thick]  (\x-\b+\a*\y-\a*\b-\sc,\y-\b) [rounded corners = 6 pt] -- (\x-\b+\a*\y+\a*\sc+\a*\b-\sc,\y+\b+\sc) [sharp corners] -- (\x+\b+\a*\y+\a*\sc+\a*\b,\y+\b+\sc);
                }            
        \node[] at (-3-3*\a,-3) {$A$};
        \node[] at (3+3*\a,3) {$B$};        
    \end{tikzpicture}
    \caption{Left picture: The (internal part of the) Riemann surface obtained by gluing two copies of the plane with holes. The elementary block is depicted in the center of the picture. Right picture: The average of two on-site observables $\lal\CO_A \CO_B \ral_{\Psi_{\Gamma}}$ of the lattice system is given by the partition function on the Riemann surface with insertions at the corresponding holes divided by the partition function without the insertions.}
    \label{fig:Network}
\end{figure}
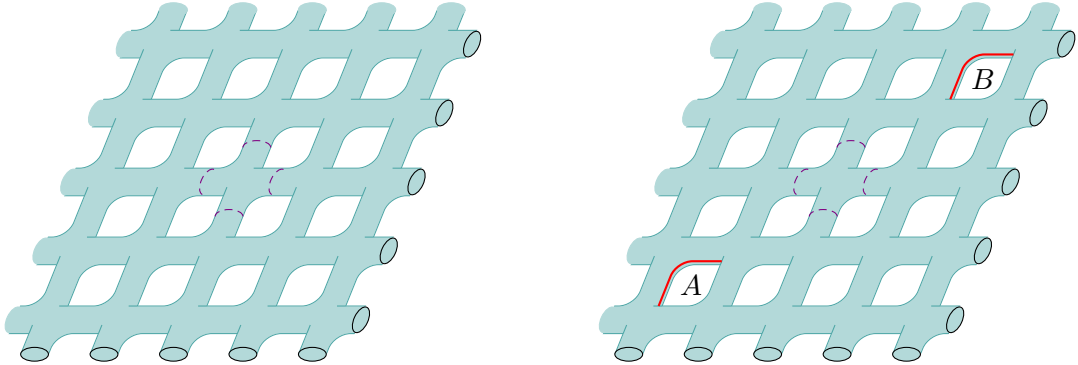

Suppose we have a finite square lattice $\Gamma$ on the complex plane. To each site $j \in \Gamma$, we assign an infinite-dimensional Hilbert space $\CV_j$ isomorphic to the Hilbert space $\CV$ of the vacuum module of a good\footnote{By good, we mean a unitary rational vertex operator algebra satisfying conditions from \cite{HuangRVOA} guaranteeing that the category of representations of such an algebra is modular.} unitary rational vertex operator algebra $\voa$. We imagine a disk-like hole at each site $j \in \Gamma$ and interpret $\CV_j$ as the Hilbert space for chiral modes living on the boundary of the hole. The vertex operator algebra defines the evaluation map $\lal \Psi_{\Gamma}| : \bigotimes_{j \in \Gamma} \CV_j \to \CCC$, that gives the amplitude for the Riemann surface with fixed state vectors on the boundary. More explicitly, if $|e_{m_{\Gamma}} \ral = \bigotimes_{j \in \Gamma} |e_{m_j}\ral$ for orthonormal basis vectors $|e_{m_j}\ral \in \CV_j$, then
\beq
\lal \Psi_{\Gamma}|e_{m_{\Gamma}} \ral = \lal\prod_{j \in \Gamma} \CO_{\emb_j}(e_{m_j})\ral
\eeq
where $\CO_{\emb_j}(e_{m_j})$ is the vertex operator producing $|e_{m_j}\ral$ on the boundary of the hole at $j \in \Gamma$ (which we define explicitly below). This map gives a state vector $|\Psi_{\Gamma}\ral \in \bigotimes_{j \in \Gamma} \CV_j$ for our lattice system on $\Gamma$. 

The norm $\lal \Psi_{\Gamma}|\Psi_{\Gamma}\ral$ is given by the partition function on the Riemann surface, obtained by gluing two copies of the plane with holes along the boundary, in the appropriate sector (see Fig. \ref{fig:Network}). The sector is specified uniquely by the requirement that only the states of the trivial module of $\voa$ are running through the tubes connecting the holes. We can cut the surface into elementary blocks, as shown in the picture, and represent $\lal \Psi_{\Gamma}|\Psi_{\Gamma}\ral$ by the partition function of a tensor network statistical model (together with some boundary conditions), with tensors defined by the amplitudes of elementary blocks for various choices of state vectors on its four boundaries. The averages $\lal \CO \ral_{\Psi_{\Gamma}} = \lal \Psi_{\Gamma}|\CO|\Psi_{\Gamma}\ral / \lal \Psi_{\Gamma}|\Psi_{\Gamma}\ral$ of bounded local observables of the lattice system can be expressed through the averages of products of on-site observables. The average $\lal \CO_{j_1} ... \CO_{j_n} \ral_{\Psi_{\Gamma}}$ for on-site observables $\{\CO_{j_a}\}_{a = 1,...,n}$, $j_a \in \Gamma$ is computed by the normalized partition function on the same Riemann surface with proper insertions along the boundaries of the holes corresponding to $j_a$ (see Fig. \ref{fig:Network}). When the size of the holes is large, the necks of the elementary blocks are thin, and the contribution of states with non-zero conformal weight running through the necks is highly suppressed. This allows us to treat the model analytically if the holes are large enough and show the existence of the thermodynamic limit. It also allows us to show that the correlations between local lattice observables in the bulk (i.e. localized far away from the boundary) decay exponentially with the distance between their supports. The correlation length is defined by the size of the holes (or the thickness of the necks). Note that if it were not for the holes all around the observables, the correlations would decay only polynomially. In particular, they do decay polynomially for observables near the boundary of the lattice. For bulk observables, the holes effectively "screen" the correlations.

While $|\Psi_{\Gamma}\ral$ depends on the positions and the shapes of the holes, we will argue that various choices give states in the same topological phase (i.e. they are related by a local Hamiltonian evolution \cite{hastings2005quasiadiabatic}) as long as the holes homogeneously fill the plane, so that there is a uniform bound on the thickness of the necks in the bulk. In particular, states with additional holes can be produced by entangling with ancillas using almost local unitaries.

Before describing the construction in more detail, let us illustrate some properties of the state for $\voa$ being the algebra of a free periodic chiral boson with the $U(1)$ current algebra at level $k = 1/m \in 1/(2 \ZZ)$ as a subalgebra generated by $J(z) = i \p \phi$. The lattice system has the induced on-site $U(1)$ action, and the state produced by the construction is $U(1)$-invariant. The action of an on-site charge operator $Q_j$ on $|\Psi_{\Gamma}\ral$ is characterized by
\beq
\lal \Psi_{\Gamma} | Q_{j} | e_{m_{\Gamma}}\ral = \lal \l \oint_{C_j} \frac{dz}{2 \pi i} J(z) \r \prod_{j \in \Gamma} \holeO_{\emb_j}(e_{m_j}) \ral
\eeq
where $C_j$ is a loop around the hole at site $j \in \Gamma$. For a region $A$, the action of the operator $Q_A = \sum_{j \in A} Q_j$ of the $U(1)$ charge inside this region corresponds to the insertion of the integrals of the current $J(z)$ along $C_j$ for each $j \in A$ and therefore (see fig. \ref{fig:u1action})
\beq
\lal \Psi_{\Gamma} | Q_{A} | e_{m_{\Gamma}}\ral = \lal \l \oint_{\p A} \frac{dz}{2 \pi i} J(z) \r \prod_{j \in \Gamma} \holeO_{\emb_j}(e_{m_j}) \ral.
\eeq
As we argue below, the perturbation of the state $\Psi_{\Gamma}$ by $Q_A$ can be performed by the action of a Hamiltonian $K_{A \bar{A}}$ which is a sum of bounded almost local terms localized near the boundary of $A$. States having this property are known to possess a topological invariant\footnote{Strictly speaking, the invariant is only defined in the thermodynamic limit. But it in this section we give an informal treatment.} \cite{bachmann2019many,ThoulessHall} defined by $\eqberry := 2 \pi i \lal[K_{X\bar{X}}, K_{Y \bar{Y}}]\ral$, where $X$ and $Y$ can be chosen to be the right and the upper half-planes. It does not depend on the position of $X$ and $Y$, and when the state is the ground state of some gapped Hamiltonian, it coincides with the zero-temperature Hall conductance. The topological invariant for $\Psi_{\Gamma}$ is given by
\beq
\eqberry = 2 \pi i \left[ \int_{\text{Im} z = 0} \frac{dz}{2 \pi i}, \int_{\text{Re} w = 0} \frac{d w}{2 \pi i} \right] \lal J(z) J(w) \ral
\eeq
where the average is taken on a doubled Riemann surface in the same sector as in the computation of $\lal \Psi_{\Gamma}|\Psi_{\Gamma}\ral$. Since only the singular part of the operator product expansion $J(z) J(w) = k (z-w)^{-2} + ...$  contributes to the integral, we have $\eqberry = k$.

It is easy to implement the insertion of anyons into the construction. The corresponding states can be expressed in terms of the correlation functions of vertex operators in non-trivial modules of $\voa$. In the present case, anyons are labeled by $p=0,...,m-1$. Schematically, for a single anyon we have
\beq
\lal \Xi^{(p)}_{\Gamma}| e_{m_\Gamma} \ral \propto \lal e^{- i p \phi(\infty)} \l\prod_{j \in \Gamma} \holeO_{\emb_j}(e_{m_j}) \r e^{i p \phi(z)} \ral
\eeq
where $p$ and $z$ are the label and the position of the anyon, and $e^{- i p \phi(\infty)}$ is inserted to satisfy the superselection rules. The charge of the anyon can be easily computed as the action of the total charge operator on $| \Xi^{(p)}_{\Gamma} \ral$ differs from the action on $|\Psi_{\Gamma}\ral$ by the insertion of $\oint J(z) \frac{dz}{2 \pi i}$ along a small loop around $z$ that simply multiplies $|\Xi^{(p)}_{\Gamma}\ral$ by $p/m$. In the same way, we can produce states with anyons for a general vertex algebra.

\section{General construction}  \label{sec:general}

To formalize the description of the previous section, we first introduce some notation and terminology.

\subsection{Preliminaries}   

\subsubsection{Vertex operator algebra}

Let $\voa$ be a good unitary rational vertex operator algebra with the vacuum vector $|0\ral \in \voa$ and vertex operators $Y(\cdot,z):\voa \to \text{End}\,\voa[[z^{\pm1}]]$. The energy-momentum tensor (the vertex operator for the conformal element) is denoted $T(z) = \sum_{n \in \ZZ} L_n z^{-n-2}$ with the central charge $c$. Let $\text{Rep}(\voa)$ be the category of representations of $\voa$, which has the structure of a unitary modular tensor category. We denote the finite set of simple objects of $\text{Rep}(\voa)$ by $\ObRepV$ with the trivial object denoted $a={\bf 1}$ and with the dual label being $\bar{a}$. We let $\voa^{(a)}$ be the corresponding modules, and let $\CV^{(a)}$ be the corresponding Hilbert spaces. We choose an orthonormal basis $e^{(a)}_m$ in $\voa^{(a)}$ and the corresponding basis $|e^{(a)}_m \ral$ in $\CV^{(a)}$ labeled by integers $m \in \NN_0$ such that $|e^{(a)}_m \ral$ are eigenvectors of $L_0$ with the weight $h^{(a)}(m)$ and $h^{(a)}(m) \leq h^{(a)}(m')$ if $m < m'$. When we omit the label $a$, we mean $a = \bf{1}$.

By the state-operator correspondence, any vector $|e_{m} \ral \in \CV$ on the boundary of the unit disk $D$ can be obtained by the insertion of $Y(e_m,0)$. More generally, if we have a holomorphic embedding of the unit disk $\emb:D \to \CCC$, in order to obtain a state $|e_m\ral$ on the boundary of the closure of the image, we need to insert (see e.g. Section 5.4 of \cite{FrenkelBenZvi})
\beq
\holeO_{\emb}(e_m) : = Y(R_{\emb} e_m, \emb(0))
\eeq
where
\beq \label{eq:Rphi}
R_\emb = v_0^{L_0} \exp \biggl( \sum_{m=1}^{\infty} v_m L_m \biggr) = \exp \biggl( \sum_{m=1}^{\infty} \frac{v_m}{v_0^m} L_m \biggr) v_0^{L_0}
\eeq
and $v_n \in \CCC$ are given by the solution of
\beq
v_0 \exp \biggl( \sum_{m=1}^{\infty} v_m z^{m+1} \p_z \biggr) z = \emb(z) - \emb(0).
\eeq
We also introduce $\lal \alpha, \emb|$ as the covector that is parallel to $\lal0|\holeO_\emb(\alpha)^{\dagger}$ and normalized such that $\lal \alpha,\emb |\holeO_\emb(\alpha) |0\ral = 1$.

\subsubsection{Evaluation maps} \label{ssec:evalmaps}

Let $\Gamma$ be a finite set, and let $\CV_{\Gamma} := \bigotimes_{j \in \Gamma} \CV_j$ where $\CV_j \cong \CV$. Let $\emb_{\Gamma}:= \{\emb_j\}_{j \in \Gamma}$ be a collection of holomorphic embeddings $\emb_j: D \to \CCC$ of open unit disks with disjoint images, which closures we denote by $B_j$. The vector $|\Psi_{\Gamma} \ral \in \CV_{\Gamma}$ is defined by 
\beq \label{eq:chiralstate1}
\lal \Psi_{\Gamma} |e_{m_\Gamma}\ral = \lal \prod_{j \in \Gamma} \holeO_{\emb_j}(e_{m_j}) \ral
\eeq
where $|e_{m_\Gamma} \ral := \bigotimes_{j \in \Gamma} |e_{m_j}\ral \in \CV_{\Gamma}$. This vector is well-defined since $\lal \Psi_{\Gamma}| \Psi_{\Gamma} \ral$ is given by the partition function on the Riemann surface obtained by gluing two copies of $\CCC\backslash \{B_j\}_{j \in \Gamma}$ along the boundaries of $B_j$ (in the sector with only the states of the trivial module running through the tubes connecting the holes) and therefore is finite. We denote the corresponding pure state on the algebra $\BoundedOps(\CV_{\Gamma})$ of bounded operators on $\CV_{\Gamma}$ by $\Psi_{\Gamma}:\BoundedOps(\CV_{\Gamma}) \to \CCC$. We can interpret $\CV_{\Gamma}$ as the Hilbert space for chiral modes living on the boundaries of $B_j$.

Let $\Gamma_{*}$ be a pointed set $\Gamma \sqcup \{*\}$ obtained by adjoining a new element $*$, and let $\CV_{\Gamma_*} := \CV \otimes \CV_{\Gamma}$. If, in addition to the data from the previous paragraph, we have a holomorphic embedding $\emb_*:D \to \CCC$ such that its image contains all $\{B_j\}_{j \in \Gamma}$, we can define a vector $|\Psi_{\Gamma_*}\ral \in \CV_{\Gamma_*}$ by
\beq \label{eq:chiralstate2}
\lal \Psi_{\Gamma_*} |e_{m_*} e_{m_\Gamma} \ral = \lal e_{m_*}, \emb_*| \prod_{j \in \Gamma} \holeO_{\emb_j}(e_{m_j})|0 \ral.
\eeq
This vector defines a natural pure state $\Psi_{\Gamma_*}$ of modes on the boundary of $\Sigma$ that is the closure of $(\Im \emb_{*}) \backslash \l \bigcup_{j \in \Gamma} B_j \r$.

\subsubsection{Conformal field theory}

When we say CFT, we always mean the diagonal (or Cardy case \cite{fuchs2002tft}) conformal field theory associated with $\voa$. We denote the anti-holomorphic sector by $\overline{\voa}$. The Hilbert space on a circle is $\CH^{\text{CFT}} = \bigoplus_{a \in \ObRepV} \l \CV^{(a)} \otimes \overline{\CV}^{(a)} \r$. The topological line defects and the elementary conformal boundary conditions \cite{cardy1989boundary} for the diagonal CFT are labeled by elements of $\ObRepV$. We denote the Hilbert space of states on the interval with the left boundary condition $a$ and the right boundary condition $b$ by $\CH^{(ab)}$. We choose a basis $|e^{(ab)}_m \ral$, $m \in \NN_0$ ordered by weights $h^{(ab)}(m)$. We also use $\Hleg := \bigoplus_{a,b \in I} \CH^{(ab)}$.

An elementary conformal boundary condition on the outer part of the disk of radius $r<1$ corresponds to the state on the circle
\beq
|a \ral = r^{-c/6} \sum_{b \in \ObRepV} \frac{S_{ab}}{\sqrt{S_{{\bf 1} b}}} \sum_{m = 0}^{\infty} r^{2h^{(b)}(m)} |e^{(b)}_m \ral \otimes |\overline{e^{(b)}_m} \ral \in \CH^{\text{CFT}}
\eeq
where $r^{-c/6}$ is the Liouville factor for the disk, and $S_{ab}$ are the components of the $S$ matrix. Their linear combinations
\beq
|b \ral\!\ral = (S_{{\bf 1}b})^{1/2} \sum_{a \in \ObRepV} S^*_{ba} |a\ral = r^{-c/6} \sum_{m = 0}^{\infty} r^{2h^{(b)}(m)} |e^{(b)}_m \ral \otimes |\overline{e^{(b)}_m} \ral
\eeq
are called the {\it Ishibashi states}. We extensively use the linear combination of elementary boundary conditions (or topological line defects) that corresponds to the Ishibashi state of the vacuum
\beq
|{\bf 1} \ral\!\ral = (S_{{\bf 1 1}})^{1/2} \sum_{a \in \ObRepV} S_{{\bf 1}a} |a\ral = r^{-c/6} \sum_{m = 0}^{\infty} r^{2h(m)} |e_m \ral \otimes |\overline{e_m} \ral
\eeq
We call it  {\it cloaking linear combination} following \cite{Runkel21}, where its various properties are discussed. In particular, any topological line defect can freely pass through a hole with the cloaking linear combination of boundary conditions without changing the correlation functions
\beq
    \begin{tikzpicture}[baseline={([yshift=-.5ex]current bounding box.center)}]
        \filldraw[color = teal!30] (0,0) circle (1);
        \filldraw[color = blue, fill = white] (0,0) circle (.5);
        \draw[color = violet, thick] (-1,0) .. controls (-0.5,0) and (-0.5,0.65) .. (0,0.65) .. controls (0.5,0.65) and (0.5,0) .. (1,0);    
    \end{tikzpicture}
    =
    \begin{tikzpicture}[baseline={([yshift=-.5ex]current bounding box.center)}]
        \filldraw[color = teal!30] (0,0) circle (1);
        \filldraw[color = blue, fill = white] (0,0) circle (.5);
        \draw[color = violet, thick] (-1,0) .. controls (-0.5,0) and (-0.5,-0.65) .. (0,-0.65) .. controls (0.5,-0.65) and (0.5,0) .. (1,0);            
    \end{tikzpicture}
    .
\eeq

The discussion about the state-operator correspondence from the previous section can be generalized to bulk and boundary field insertions of the diagonal CFT. For a holomorphic embedding $f:D \to \CCC$ of the unit disk, we denote the bulk field insertion at $f(0)$ that produces $|\alpha\ral \otimes |\bar{\alpha}\ral \in \CV^{(a)} \otimes \overline{\CV}^{(a)}$ on the boundary of the image by $\holeO^{(a)}_{f}(\alpha \otimes \bar {\alpha})$. For $f$ that maps the interval $[-1,0]$ to the boundary with the elementary boundary condition $a$, the interval $[0,1]$ to the boundary with the elementary boundary condition $b$ and the upper half-disk to the bulk, the boundary field insertion at $f(0)$ that produces an open state vector $|\alpha\ral \in \CH^{(ab)}$ on the image of the upper half-circle is denoted $\holeO^{(ab)}_{f}(\alpha)$.

\subsection{Chiral state}
\label{ssec:chiralstate}

Let $\Lambda$ be the two-dimensional lattice which we identify with $\ZZ^2 \subset \CCC$. The elements of the lattice are called sites, while the segments connecting two neighboring sites are called edges. The dual lattice is denoted $\Lambda^{\vee}$. The location of $j\in \Lambda$ or $j \in \Lambda^{\vee}$ on the complex plane is denoted $\siteloc_j \in \CCC$. For a finite subset $\Gamma \subset \Lambda$ we let the algebra of observables of the finite lattice system on $\Gamma$ be $\SA_{\Gamma} := \bigotimes_{j \in \Gamma} \CB(\CV_{j})$. When $\voa$ is fermionic, we use the graded tensor product. In the thermodynamic limit, one can define the algebra of quasi-local observables $\SA$ on $\Lambda$, which is the norm completion of the algebra of local observables $\SAl$ defined by the canonical direct limit $\SAl := \underset{\Gamma}\varinjlim\, \SA_{\Gamma}$.

We consider $\Gamma \subset \Lambda$ that are the sets of sites inside the square $\{z \in \CCC: |\Re(z)|, |\Im(z)| < N+1/2\}$ for some $N \in \NN$. As explained in Section \ref{ssec:evalmaps}, to each such subset we can associate vectors $|\PsiNB_{\Gamma}\ral \in \CV_{\Gamma}$ and $|\Psi_{\Gamma_*} \ral \in \CV_{\Gamma_*}$ if we choose a collection of holomorphic embeddings $\emb_{\Gamma_*} := \{\emb_j\}_{j \in \Gamma_*}$ of the unit disk. To fix this data, we choose a function $f(z)$ that maps the unit disk into the interior of the square with vertices at $(\pm1 \pm i) \eps$ for $0<\eps<1/2$ and let $\emb_j(z):=f(z)+\siteloc_j$. We also fix some $\emb_*$ such that $\Im \emb_*$ is given by the square from the first sentence of this paragraph. The resulting surface $\Sigma$ is the square $\Im \emb_*$ with a collection of holes at each site of $\Gamma$ (see Fig. \ref{fig:sigmatau}).

In the following, for convenience, we choose a particular one-parameter family of functions $f(z)$ defined in eq. (\ref{eq:holeshape}) with the parameter $\tau > 0$ that characterizes the size of the holes. When $\tau \to 0$, the holes almost touch each other, while when $\tau$ is large, the holes are very small. To emphasize the dependence on $\tau$, we write $\Sigma_{\tau}$, $\Psi_{\tau,\Gamma}$ and $\Psi_{\tau,\Gamma_*}$

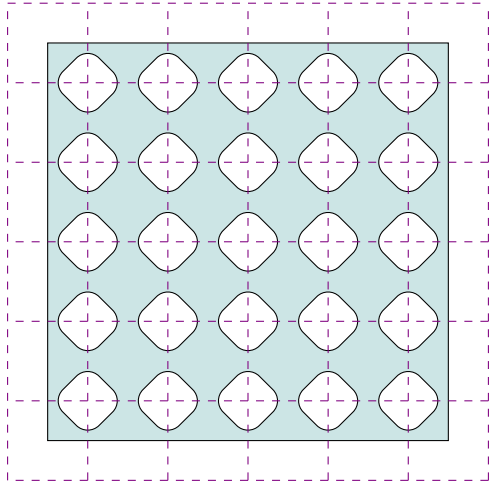
\begin{figure}
    \centering
    \begin{tikzpicture} [x=15pt,y=15pt]
        \draw [black, fill=teal, fill opacity = 0.2] (-5,-5) rectangle (5,5);
        \draw [violet, dashed] (-6,-6) -- (-6,6);
        \draw [violet, dashed] (-6,6) -- (6,6);
        \foreach \x in {-4,-2,...,4}
            \foreach \y in {-4,-2,...,4}
                {
                \draw [black, fill=white, rounded corners=6pt, rotate around={45:(\x,\y)}]  (\x-.9/1.41,\y-.9/1.41) rectangle (\x +.9/1.41,\y +.9/1.41);
                }
        \foreach \x in {-4,-2,...,6}
            \foreach \y in {-6,-4,...,4}
                {
                \draw [violet, dashed] (\x,\y) -- (\x-2,\y);
                \draw [violet, dashed] (\x,\y) -- (\x,\y+2);        
                }
    \end{tikzpicture}
    \caption{The surface $\Sigma_{\tau}$ and its partition into elementary blocks.}
    \label{fig:sigmatau}
\end{figure}

The norm $\lal \Psi_{\tau,\Gamma_*}|\Psi_{\tau,\Gamma_*} \ral$ is given by the partition function on the Riemann surface obtained by gluing two copies of $\Sigma_{\tau}$ in the sector with only the states of the trivial module running through the tubes connecting the holes. Alternatively, it is given by the partition function of CFT on $\Sigma_{\tau}$ with the cloaking linear combinations of elementary boundary conditions. It can be computed by cutting the resulting surface along the edges of $\Lambda$ and gluing using the usual rules (see Fig. \ref{fig:sigmatau}). Such a decomposition defines a tensor network statistical model that is similar to the one introduced in \cite{Runkel21}. Note that when the holes are large, the necks of the elementary blocks are very thin, and therefore the states with non-zero weight running through the necks are suppressed.

We claim that the direct limit of the restriction of the state $\Psi_{\tau,\Gamma_*}$ to $\SA_{\Gamma}$ over $\Gamma \subset \Lambda$ exists and defines a pure state $\Psi_{\tau}$ on $\SA$. In Appendix \ref{app:clusterCFT}, we argue that it is true at least for sufficiently small but finite $\tau$ using the cluster expansion, though we believe that it is true for any $\tau$. We also show that the state $\Psi_{\tau}$ has exponential decay of correlations of local observables.

\begin{remark}
While we have fixed the lattice to be a square lattice and the shape of the holes, we believe that construction works for any (not necessarily ordered) lattice and holes around sites that homogeneously fill the plane so that the widths of the necks of all elementary blocks are uniformly bounded. In particular, for the cluster expansion to work, we only need all the holes to be relatively close to each other so that all the necks are sufficiently thin.
\end{remark}

\begin{remark} \label{rmk:nodisentangling}
We can make the correlation length arbitrarily small by choosing a small enough value of $\tau$. When $\tau$ grows, the holes shrink, and each site becomes more and more disentangled. However, the correlation length grows with $\tau$ since the necks of the elementary blocks become thicker and the states running through the necks are becoming less and less suppressed. Therefore, decreasing the size of the holes homogeneously does not allow to disentangle the sites without destroying the locality.
\end{remark}

\begin{remark}
The states $\Psi_{\tau,\Gamma}$ and $\Psi_{\tau, \Gamma_*}$ almost coincide in the bulk, and therefore we can also get $\Psi_{\tau}$ by taking the thermodynamic limit of $\PsiNB_{\tau,\Gamma}$. The advantage of working with $\Psi_{\tau,\Gamma_*}$ over $\PsiNB_{\tau,\Gamma}$ is that the former has decaying correlations for any two distant observables of $\SA_{\Gamma}$, while the latter has relatively large correlations between distant observables near the boundary (as the decay of such correlations is only polynomial). The entanglement with the modes living in the auxiliary factor $\CV$ in $\CV_{\Gamma_*}$ effectively localizes the boundary correlations.
\end{remark}

\begin{remark}
One can analogously define states on a finite lattice modeling the system on an arbitrary Riemann surface. In this case, one generally has a non-trivial space of conformal blocks, and each basis element gives its own state. 
\end{remark}

In the remainder of the paper, we give arguments in support of the fact that $\Psi_{\tau}$ has the topological order associated with the vertex algebra $\voa$.

\subsubsection{Local perturbations} \label{ssec:perturb}

One natural way to modify the state $\Psi_{\tau}$ locally is to insert a vertex operator into the defining correlator eq.(\ref{eq:chiralstate2}) for $\Psi_{\Gamma_*}$ (such that the resulting vector in $\CV_{\Gamma_*}$ is non-zero) and take the thermodynamic limit. Such an insertion corresponds to a change of an element of the tensor network statistical model that computes the averages of observables in the state $\Psi_{\tau}$. The analysis from Appendix \ref{app:clusterCFT} implies that the change of the expectation values of observables outside a large disk around the insertion exponentially decays with the size of the disk. Using this, one can argue that the modified state can be obtained from $\Psi_{\tau}$ by an almost local\footnote{By almost local, we mean that the observable can be approximated by a local one with the error (in the operator norm) that decays rapidly (faster than any power) with the diameter of the support (see \cite{LocalNoether} for a precise characterization).} unitary observable $\CU$ (see Appendix \ref{app:localpert}), i.e. the average of $\CA \in \SA$ in the modified state is given by $\lal \CU^* \CA \CU\ral_{\Psi_{\tau}}$. 

A local variation of the moduli of the Riemann surface $\Sigma_{\tau}$ (e.g. changing the size of a hole) corresponds to the insertion of a certain combination of vertex operators from the Virasoro subalgebra, and therefore can also be performed by an almost local unitary observable. The localization of the observable (i.e. the rate of decay of the "tails") depends on the correlation length. Suppose a global variation of the moduli is composed of local changes which keep the correlation length bounded. Using unitaries for local changes, one can construct a finite time Hamiltonian evolution, with the Hamiltonian being a sum of almost local terms, that performs the global change. Therefore, states related by such a global variation are in the same topological phase. In particular, the states $\Psi_{\tau}$ are in the same phase for different values of $\tau$ as long as the correlation length stays finite.

One can also modify the state by entangling it with ancillas. If we describe the initial ancillary state as a pure state of the trivial module living on the boundary of a decoupled disk, one can create a "wormhole" between this disk and $\Sigma_{\tau}$ (i.e. cutting a small disk out of each surface and gluing the boundaries) by inserting a certain combination of vertex operators as only the states of the trivial module are running through this wormhole. Hence, the modified state can be obtained from the original one by an almost local unitary which entangles the sites of the lattice with the ancilla. By inverting the process, we can disentangle any single site of the lattice system. Note, however, that in this way, we can not disentangle all the spins of the system keeping the correlation length finite (see Remark \ref{rmk:nodisentangling}).

\subsubsection{Anyons} \label{ssec:anyons}


One can also modify the state $\Psi_{\tau}$ by inserting non-trivial modules of $\voa$. A choice of a holomorphic embedding of the open unit disk $\anyonemb:D \to \Sigma_{\tau}$ and a vector $\alpha^{(a)} \in \voa^{(a)}$ gives the map $\CV^{(a)} \otimes \CV_{\Gamma} \to \CCC$ that evaluates the amplitude on the Riemann surface with fixed state vectors $|\alpha^{(a)}\ral \in \CV^{(a)}$, $|e^{(a)}_{m_*}\ral \in \CV^{(a)}$, $\{|e_{m_j}\ral \in \CV_j\}_{j \in \Gamma}$ on the boundaries of the closures of the images of $\anyonemb$, $\emb_*$, $\{\emb_j\}_{j \in \Gamma}$, respectively. Therefore $\anyonemb$ and $\alpha^{(a)}$ define $|\Xi_{\tau, \Gamma_*}^{(a)} \ral  \in \CV^{(a)}_{\Gamma_*} := \CV^{(a)} \otimes  \CV_{\Gamma}$. We interpret the thermodynamic limit $\Xi_{\tau}^{(a)}$ of the corresponding state on $\SA_{\Gamma}$ as a state of the anyon of type $a$ localized on sites in the vicinity of the image of $\anyonemb$. More generally, one can insert several anyons by choosing a collection of embeddings of open unit disks into $\Sigma_{\tau}$, a collection of elements of the modules, and an element of the corresponding space of conformal blocks. The determination of braiding properties is reduced to the analysis of the correlation functions of CFT.

By the arguments from the previous subsection, one can move anyons (or create particle-antiparticle pairs) on the lattice using a Hamiltonian evolution, with the Hamiltonian being a sum of almost local terms localized along the path. In particular, a non-trivial anyon inside a large ring can be detected by an observable of the algebra that performs a transport of an anyon with non-trivial braiding around the ring. It implies that $\Xi_{\tau}^{(a)}$ is not related to $\Psi_{\tau}$ by an almost local unitary observable.

\subsection{Doubled state and invertibility} \label{sec:double}

While the averages in the state $\Psi_{\tau}$ can be computed using an auxiliary tensor network statistical model, the state itself is not defined as a tensor network state. However, there is a natural family of tensor network states $\Phi_{\tau}(\sigma)$, $\sigma \in (0, \infty)$ of the doubled system, such that in the limit $\sigma \to \infty$ one gets $\Psi_{\tau} \otimes \overline{\Psi}_{\tau}$, where $\overline{\Psi}_{\tau}$ is the state obtained by complex conjugation.

Let us first define $\Phi_{\tau}(\sigma)$ using the correlation functions of CFT operators. Let $\tilde{\Sigma}_{\sigma}$ be the surface obtained by taking the closure of $\Im \emb_{*}$ after removing the images of $\tilde{f}(z) + \siteloc_{k}$ for each $k \in \Lambda^{\vee}$ (see Fig. \ref{fig:sigmatausigma}), where $\tilde{f}(z)$ is given by $f(z)$ with $\tau$ replaced by $\sigma$. We define $|\Phi_{\tau,\Gamma}(\sigma) \ral \in \CV_{\Gamma} \otimes \overline{\CV}_{\Gamma}$ by
\beq \label{eq:doubledstate}
\lal \Phi_{\tau, \Gamma}(\sigma)| e_{m_{\Gamma}} \bar{e}_{m_{\Gamma}} \ral = \lal \prod_{j \in \Gamma} \holeO_{\emb_j}(e_{m_j} \otimes \bar{e}_{m_j}) \ral^{\text{CFT}}_{\tilde{\Sigma}_{\sigma}}
\eeq
with the cloaking linear combination of boundary conditions along the boundary of $\tilde{\Sigma}_{\sigma}$. Alternatively, one can define the components of $|\Phi_{\tau,\Gamma}(\sigma) \ral$ as a correlator of the vertex algebra on a Riemann surface obtained by gluing two copies of $\tilde{\Sigma}_{\sigma}$ in the sector with only the states of the trivial module running through the holes.

We claim that the direct limit of the state $\Phi_{\tau,\Gamma}(\sigma)$ over $\Gamma \subset \Lambda$ exists and defines a pure state $\Phi_{\tau}(\sigma)$ on $\SA \otimes \SA$ with exponential decay of correlations. In the same way as for $\Psi_{\tau}$, one can show that this claim holds at least for small enough $\tau$ using the cluster expansion.

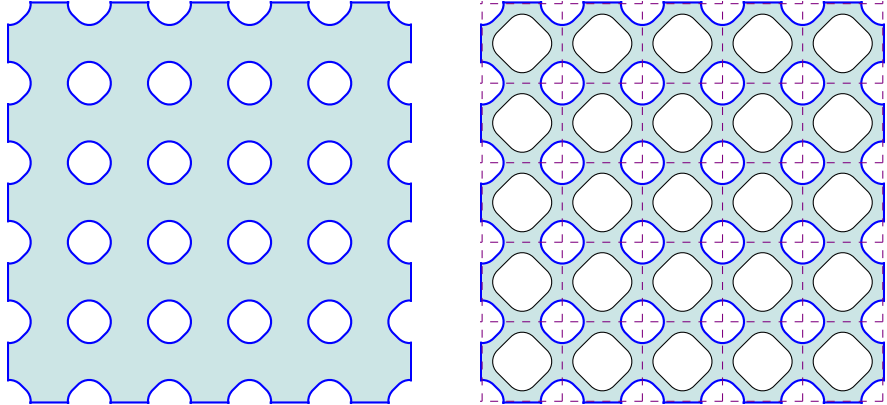
\begin{figure}
    \centering
        \begin{tikzpicture} [x=15pt,y=15pt]
    \clip (-5.05,-5.05) rectangle (5.05,5.05);
    \filldraw[color=teal, opacity = 0.2] (-5,-5) rectangle (5,5);
        \foreach \x in {-5,-3,...,5}
            \foreach \y in {-5,-3,...,5}
                \filldraw [rounded corners=6pt, rotate around={45:(\x,\y)},color = blue,fill=white, thick]  (\x-.7/1.41,\y-.7/1.41) rectangle (\x + .7/1.41,\y + .7/1.41);
        \foreach \x\y in {-4/5,-2/5,0/5,2/5,4/5}
            {
            \draw[blue,thick] (\x-0.5,\y+0.03) -- (\x+0.5,\y+0.03);
            \draw[blue,thick] (\x-0.5,-\y-0.03) -- (\x+0.5,-\y-0.03);
            \draw[blue,thick] (\y+0.03,\x-0.5) -- (\y+0.03,\x+0.5);
            \draw[blue,thick] (-\y-0.03,\x-0.5) -- (-\y-0.03,\x+0.5);
            }
    \end{tikzpicture}
    \qquad 
    \begin{tikzpicture} [x=15pt,y=15pt]
    \clip (-5.05,-5.05) rectangle (5.05,5.05);
    \filldraw[color=teal, opacity = 0.2] (-5,-5) rectangle (5,5);
        \foreach \x in {-4,-2,...,4}
            \foreach \y in {-4,-2,...,4}
                \draw [rounded corners=6pt, rotate around={45:(\x,\y)}, fill=white]  (\x-.9/1.41,\y-.9/1.41) rectangle (\x + .9/1.41,\y + .9/1.41);
        \foreach \x in {-5,-3,...,5}
            \foreach \y in {-5,-3,...,5}
                \filldraw [rounded corners=6pt, rotate around={45:(\x,\y)},color = blue,fill=white, thick]  (\x-.7/1.41,\y-.7/1.41) rectangle (\x + .7/1.41,\y + .7/1.41);
        \foreach \x\y in {-4/5,-2/5,0/5,2/5,4/5}
            {
            \draw[blue,thick] (\x-0.5,\y+0.03) -- (\x+0.5,\y+0.03);
            \draw[blue,thick] (\x-0.5,-\y-0.03) -- (\x+0.5,-\y-0.03);
            \draw[blue,thick] (\y+0.03,\x-0.5) -- (\y+0.03,\x+0.5);
            \draw[blue,thick] (-\y-0.03,\x-0.5) -- (-\y-0.03,\x+0.5);
            }
        \foreach \x in {-5,-3,...,5}
            \foreach \y in {-5,-3,...,5}
                {
                \draw [violet, dashed] (\x,\y) -- (\x-2,\y);
                \draw [violet, dashed] (\x,\y) -- (\x,\y+2);        
                }        
    \end{tikzpicture}
    \caption{The surfaces $\tilde{\Sigma}_{\sigma}$ (on the left) and $\Sigma_{\tau,\sigma}$ (on the right).}
    \label{fig:sigmatausigma}
\end{figure}

To define it as a tensor network state, let $\Sigma_{\tau,\sigma}$ be the surface obtained by removing the images of $\emb_{\Gamma}$ from $\tilde{\Sigma}_{\sigma}$ (see Fig. \ref{fig:sigmatausigma}). We can cut the surface $\Sigma_{\tau,\sigma}$ along the edges of the lattice $\Lambda^{\vee}$ to decompose it into elementary blocks. Each block gives a map $\CV \otimes \overline{\CV} \otimes \Hleg^{\otimes 4} \to \CCC$ that defines a tensor network state with the physical space being $\CV \otimes \overline{\CV}$ and with the auxiliary Hilbert space being $\Hleg$ on each leg. In Appendix \ref{app:basictensor}, we explain how one can get an expression for the components of the map in terms of the CFT correlation functions on the unit disk.

The parameter $\tau$ gives an upper bound on the correlation length of the states $\Phi_{\tau}(\sigma)$ for any $\sigma$, that follows from the auxiliary statistical model obtained by cutting along the edges of $\Lambda$ in the same way as in Section \ref{ssec:chiralstate}. When $\sigma$ is large, the presence of the holes near the sites of the dual lattice is negligible, and the state $\Phi_{\tau}(\sigma)$ is locally close to the tensor product of pure states $\Psi_{\tau}$ and $\overline{\Psi}_{\tau}$. On the contrary, when $\sigma \to 0$, the holes are large, and the individual spins become more and more disentangled. One can produce $\Phi_{\tau}(\sigma)$ from $\Psi_{\tau} \otimes \overline{\Psi}_{\tau}$ by creating "wormholes" (with only the states of the trivial module running though it) at each site of the dual lattice. As argued in Section \ref{ssec:perturb}, each wormhole can be created by an almost local unitary. The localization of such unitaries would not be affected if some wormholes have already been created. Hence, one can compose such creation processes to argue that both states are in the same topological phase. The same argument doesn't work if one tries to disentangle the state $\Phi_{\tau}(\sigma)$ by disconnecting the necks of the tensor network since the states of non-trivial modules are running through them.

The situation is special when $\voa$ is a holomorphic vertex operator algebra, i.e. $\ObRepV = \{\bf{1}\}$.
In that case, there is only a single boundary condition, and only one state running through the legs of the tensor network is not suppressed in the limit $\sigma \to 0$. The state $\Phi_{\tau}(\sigma)$ can be completely disentangled into a factorized state, and since $\Phi_{\tau}(\sigma)$ is in the same phase as $\Psi_{\tau} \otimes \overline{\Psi}_{\tau}$, the state $\Psi_{\tau}$ is invertible with the inverse being $\overline{\Psi}_{\tau}$. We conjecture that the states $\Psi_{\tau}$ for holomorphic vertex algebras with different values of $c$ are in different invertible phase.

\begin{remark}
It is believed that invertible bosonic phases in two dimensions are classified by $c \in 8 \ZZ$ \cite{Kitaevlecture}. Our construction produces a candidate for a state in a non-trivial invertible phase for any holomorphic vertex operator algebra, including a representative of Kitaev's $E_8$ phase, which is believed to be a generator of all invertible phases. It would be desirable to show that any holomorphic vertex operator algebra gives a state in the same topological phase as a stack of several $E_8$ states.
\end{remark}

\begin{remark}
While in this note we do not provide explicit parent Hamiltonians, we want to point out that given a parent Hamiltonian $H$ for a tensor network state $\Phi_{\tau}(\sigma)$ we can construct a Hamiltonian for $\Psi_{\tau}$ by applying to $H$ a unitary evolution that produces $\Phi_{\tau}(\sigma)$ from $\Psi_{\tau} \otimes \overline{\Psi}_{\tau}$ and taking the partial average of the result over $\overline{\Psi}_{\tau}$. In particular, it provides a parent gapped Hamiltonian in the invertible case (since $\Phi_{\tau}(\sigma)$ can be produced from a factorized state by a local Hamiltonian evolution and hence gapped).
\end{remark}

\subsection{Internal Lie-group symmetry}

Suppose we have a $U(1)$-invariant state of a lattice system with an on-site $U(1)$ symmetry. We say that it has no local spontaneous symmetry breaking if the perturbation by the sum of charge operators inside any region can be undone by a Hamiltonian almost localized on the boundary of that region, as explained in \cite{ThoulessHall}. More precisely, if $Q_A$ is the operator of the charge inside region $A$, then there is a Hamiltonian $K_{A \bar{A}}$ that is a sum of observables almost localized near the boundary of $A$ (with a uniform rate of decay independent of $A$), such that for any $\CA \in \SA$ we have $\lal[Q_{A},\CA]\ral_{\psi} = \lal[K_{A\bar{A}},\CA]\ral_{\psi}$. Such a 2d state has a topological invariant taking values in $H^{4}(BU(1),\RR) \cong \RR$, which  we call $U(1)$-equivariant Berry class \cite{LocalNoether}. It is locally computable and doesn't change under a local Hamiltonian evolution. When the state is a ground state of some gapped Hamiltonian, this invariant coincides with the Hall conductance. 

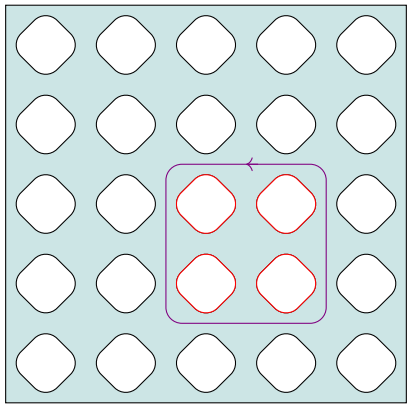
\begin{figure}
    \centering
    \begin{tikzpicture} [x=15pt,y=15pt]
        \draw [black, fill=teal, fill opacity = 0.2] (-5,-5) rectangle (5,5);
        \foreach \x in {-4,-2,...,4}
            \foreach \y in {-4,-2,...,4}
                {
                \draw [black, fill=white, rounded corners=6pt, rotate around={45:(\x,\y)}]  (\x-.9/1.41,\y-.9/1.41) rectangle (\x +.9/1.41,\y +.9/1.41);
                }
        \foreach \x in {0,2}
            \foreach \y in {0,-2}
                {
                \draw [red, fill=white, rounded corners=6pt, rotate around={45:(\x,\y)}]  (\x-.9/1.41,\y-.9/1.41) rectangle (\x +.9/1.41,\y +.9/1.41);
                }  
        \draw [violet, rounded corners=6pt]  (1-2,-1-2) rectangle (1+2,-1+2);    
        \draw [violet, <-]  (1,1) -- (1+0.3,1); 
    \end{tikzpicture}
    \caption{The action of the $U(1)$ charge on a region $A$ can be implemented by inserting $\int_{\p A} J(z) \frac{dz}{2 \pi i}$ into the defining correlator eq. (\ref{eq:chiralstate1}).}
    \label{fig:u1action}
\end{figure}

Suppose $\voa$ has a $U(1)$-symmetry at level $k$. Let $J(z)$ be the corresponding current with the operator product expansion
\beq
J(z) J(w) = \frac{k}{(z-w)^2} + ... 
\eeq
Our lattice model associated with $\voa$ naturally has the on-site $U(1)$-symmetry that corresponds to the $U(1)$ action on the vacuum module. The state $\Psi_{\tau}$ is $U(1)$-invariant. For a region $A$, let $Q_A$ be the generator of $U(1)$ action on sites inside $A$. We have
\beq
\lal \Psi_{\tau,\Gamma} | Q_{A} | e_{m_{\Gamma}}\ral = \lal \l \oint_{\p A} \frac{dz}{2 \pi i} J(z) \r \prod_{j \in \Gamma} \holeO_{\emb_j}(e_{m_j}) \ral.
\eeq
As argued in Appendix \ref{app:localpert}, there is a self-adjoint almost local observable $K(z) \in \SA$ such that the difference between the state vector for the state affected by the insertion of $J(z)$ and the vector $K(z) |\Psi_{\tau} \ral$ is proportional to $|\Psi_{\tau}\ral$. This observable can be chosen to be $U(1)$-invariant using averaging over $U(1)$ action. Since the action of $Q_A$ corresponds to the insertion of the integral of $J(z)$ along $\p A$, our state $\Psi_{\tau}$ has no local spontaneous symmetry breaking.

The topological invariant associated with $U(1)$ symmetry is given by
\beq
\eqberry = 2 \pi i \lal \left[ \int_{\text{Im} z = 0} \frac{dz}{2 \pi i}, \int_{\text{Re} w = 0} \frac{d w}{2 \pi i} \right] J(z) J(w) \ral^{\text{CFT}}_{\Sigma_{\tau}}.
\eeq
Only the singular term contributes to the integral. Since
\beq
\left[ \int_{\text{Im} z = 0} \frac{dz}{2 \pi i}, \int_{\text{Re} w = 0} \frac{dw}{2 \pi i} \right] \frac{1}{(z-w)^2} = \frac{1}{2 \pi i },
\eeq
we have $\eqberry = k$.

One can similarly consider the case of $\voa$ with $G$-symmetry for a compact Lie group $G$. Using the formalism of \cite{LocalNoether}, one can define $G$-equivariant Berry classes for $G$-invariant lattice states with no local spontaneous symmetry breaking and, in the same way as above, identify this class for the state $\Psi_{\tau}$ with the level of the $G$-current subalgebra.

\section{Examples} \label{sec:examples}

\subsection{Free fermions}

The vertex operator algebra of a single Weyl fermion is generated by the fields
\beq
\psi(z) = \sum_{n \in \ZZ} \psi_{-1/2-n} z^n,
\eeq
\beq
\bar{\psi}(z) = \sum_{n \in \ZZ} \bar{\psi}_{-1/2-n} z^n
\eeq
with $\{\bar{\psi}_n , \psi_m\} = \delta_{m+n,0}$, $\{\bar{\psi}_n,\bar{\psi}_m\} = \{\psi_n, \psi_m\} = 0$ and the operator product expansion (OPE)
\beq
\bar{\psi}(z) \psi(w) = \frac{1}{z-w} + ....
\eeq
Each site of the lattice system has infinitely many fermionic modes labeled by $r \in \ZZ+1/2$. The entanglement of modes in the state $\Psi_{\tau}$ decays with $|r|$ so that modes with large $|r|$ are almost disentangled.

We have $U(1)$ currents
\beq
J(z) = \sum_{n \in \ZZ} z^n J_{-n-1} = :\bar{\psi}(z)\psi(z):
\eeq
and the corresponding on-site $U(1)$ symmetry action for the lattice model. The $U(1)$-equivariant Berry class $\eqberry$ of $\Psi_{\tau}$ is given by the level of the $U(1)$-current algebra that, in our case, is equal to $1$.

The lattice state $\Psi_{\tau}$ is free, i.e. the average of any observable can be expressed through the two-point functions of the fermionic creation and annihilation operators using Wick's theorem. These two-point functions define a projector in the single particle Hilbert space that fully characterizes the many-body free state. The index of the projector determines whether the state is in a non-trivial phase \cite{bellissard1985, avron1994,kitaev2006anyons}. The doubled free state always has the $U(1)$ symmetry that swaps the fermionic modes, and its $\eqberry$ is given by the index of the projector of the original system. Therefore the index of the projector for $\Psi_{\tau}$ for a single Weyl fermion is $2$. 

The vertex algebra of a single Majorana-Weyl fermion is generated by
\beq
\chi(z) = \sum_{n \in \ZZ} \chi_{-1/2-n} z^n
\eeq
with $\{\chi_n , \chi_m\} = \delta_{m+n,0}$ and the OPE
\beq
 \chi(z) \chi(w) = \frac{1}{z-w} + ....
\eeq
Its double gives the Weyl fermion, and therefore the index of the projector of the state corresponding to this vertex algebra is 1.

The states $\Phi_{\tau}(\sigma)$ are also free and have the trivial index. They give a path between a factorized state and $\overline{\Psi}_{\tau} \otimes \Psi_{\tau}$ that shows invertibility of $\Psi_{\tau}$.

\subsection{Free bosons}

With every finite rank lattice $L$ equipped with a positive symmetric bilinear form $K: L \times L \to \ZZ$, one can associate the vertex operator algebra of free periodic bosons with the OPE
\beq
\phi^a(z) \phi^b(w) = - K^{ab} \log(z-w) + ...
\eeq
Depending on whether the lattice is even or odd, we can generate examples of chiral states of bosonic and fermionic lattice systems.

\subsubsection{\texorpdfstring{$E_{8}$}{Lg} state} \label{ssec:E8}

The simplest non-trivial free bosonic system that has no anyons is given by free periodic bosons $\phi^a$, $a=1,...,8$ with 
\beq
\lal \phi^a(z) \phi^b(w) \ral = -(C^{-1})^{ab} \log(z-w).
\eeq
where $(C^{-1})^{ab}$ is the inverse of the Cartan matrix of $E_8$. The corresponding state $\Psi_{\tau}$ is invertible, as discussed in Section \ref{sec:double}.

\subsubsection{\texorpdfstring{$U(1)_{m}$}{Lg} state} \label{ssec:U1}

Let us now consider the topologically ordered states produced by a single periodic boson with 
\beq
\lal \phi(z) \phi(w) \ral = - \frac{1}{m} \log(z-w)
\eeq
for $m \in 2 \ZZ$. The simple modules are labeled by $p = 0,...,m-1$. The module with label $p$ corresponds to the field $e^{i p \phi}$. 

The system has $U(1)$ symmetry $J(z) = i \p \phi$ at level $k=1/m$
\beq
J(z) J(w) = \frac{1/m}{(z-w)^2} + ...
\eeq
Therefore the resulting state realizes a fractional quantum Hall topological order with $\eqberry = 1/m$.

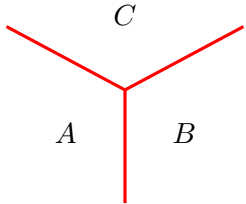
\begin{figure}
\centering
\begin{tikzpicture}[scale=.3]
\draw[red, very thick] (0,0) -- (1.5*3.4641,1.4*2);
\draw[red, very thick] (0,0) -- (-1.5*3.4641,1.4*2);
\draw[red, very thick] (0,0) -- (0,-5);

\node  at (0,1.5*2+0.3) {$C$}; 
\node  at (-1.5*1.7321,-1.5*1-0.4) {$A$};
\node  at (1.5*1.7321,-1.5*1-0.4) {$B$};

\end{tikzpicture}
\caption{The partition of the plane into three cone-like regions.}
\label{fig:ABC}
\end{figure}

For any state that has no local spontaneous $U(1)$-symmetry breaking, one can define automorphisms of the algebra $\SA$ that insert a unit of flux with the charge and the statistics defined by $\eqberry$ \cite{ThoulessHall}. Physically, such an automorphism corresponds to a finite time evolution with the Hamiltonian being a sum of almost local terms. They can be defined in the following way. Let $A$, $B$, $C$ be the partition of the plane into three cone-like regions (see Fig. \ref{fig:ABC}). The action of the charge $Q_A$ on the region $A$ can be undone by the action of some Hamiltonian that is a sum of almost local terms localized on the boundary of $A$. It can be split into two parts, $K_{AB}$ and $K_{AC}$, which are localized near the half-lines $AB$ and $AC$, respectively. The flux insertion then corresponds to the unit time evolution by the Hamiltonian $2 \pi (Q_A-K_{AB})$.

For the state $\Psi_{\tau}$, one can choose $K_{AB}$ such that the action of $(Q_A-K_{AB})$ corresponds to the insertion of $\int_{AB} \frac{dz}{2 \pi i} J(z)$ where the integral is performed along the half-line $AB$. Since $J(z) = i \p \phi$, the state with the unit of flux and the state obtained by the insertion of $e^{i \phi}$ into the triple point $ABC$ are related by an almost local unitary. Thus, both states represent the anyon with fractional charge $\eqberry = 1/m$.

\appendix

\section{Weierstrass elliptic function} \label{app:Weierstrass}

Let $\Lambda$ be the lattice $n_1 \omega_1 + n_2 \omega_2$, $n_1,n_2 \in \ZZ$ generated by $\omega_1,\omega_2 \in \CCC$. The Weierstrass elliptic function is defined by 
\beq
\wp_{\omega_1,\omega_2}(z) = \wp_{\Lambda}(z) := \frac{1}{z^2} + \sum_{\lambda \in \Lambda \backslash \{0\}} \l \frac{1}{(z-\lambda)^2} - \frac{1}{\lambda^2} \r,
\eeq
\beq
\wp'_{\omega_1,\omega_2}(z)= \wp'_{\Lambda}(z) := - \sum_{\lambda \in \Lambda } \l \frac{2}{(z-\lambda)^3} \r.
\eeq
Any elliptic function is a rational function of $\wp_{\Lambda}(z)$ and $\wp'_{\Lambda}(z)$. Some useful constants are $e_1 = \wp_{\Lambda}(\omega_1/2)$, $e_2 = \wp_{\Lambda}(\omega_2/2)$, $e_3 = \wp_{\Lambda}((\omega_1+\omega_2)/2)$ and $g_2 = -4(e_1 e_2 + e_2 e_3 + e_1 e_3) = 60 G_4$, $g_3 = 4 e_1 e_2 e_3 = 140 G_6$, where
\beq
G_k = \sum_{\lambda \in \Lambda \backslash \{0\}} \lambda^{-k}.
\eeq
We have
\beq
\wp_{\Lambda}'(z)^2 = 4 \wp_{\Lambda}(z)^3 - g_2 \wp_{\Lambda}(z) - g_3 = 4 (\wp_{\Lambda}(z)- e_1) (\wp_{\Lambda}(z)- e_2) (\wp_{\Lambda}(z)- e_3).
\eeq
The inverse of $\wp_{\Lambda}(z)$ is given by
\beq
u_{\Lambda}(w) := - \int^{\infty}_w \frac{d y}{\sqrt{4 y^3 - g_2 y - g_3}} = - \int^{\infty}_w \frac{d y}{\sqrt{4(y-e_1)(y-e_2)(y-e_3)}}.
\eeq

\section{The shape of the holes and the basic tensor} \label{app:basictensor}

Let us set $n=4$. Consider cylinders with the coordinate $w \sim w + 1$ and cuts along the half-lines $[\frac{k}{n}, \frac{k}{n} - i \infty]$, $k=0,1,2,3$. We assign such cylinders to sites of the square lattice, such that for a given site, each cut corresponds to one of the four adjacent edges, and glue them along the cuts corresponding to the same edge to produce a Riemann surface. Let $u_{\Lambda}(z)$ be the inverse of the Weierstrass function $\wp_{\Lambda}$ for the lattice $\Lambda$ and $e_1 = \wp_{\Lambda}((1+i)/ 2)$. The function $u_{\Lambda}(e_1 e^{-4 \pi i w})$ maps the Riemann surface to the plane
such that the circles with constant $\Im w$ are mapped to contours around sites of $\Lambda$ or $\Lambda^{\vee}$ depending on weather $\Im w>0$ or $\Im w<0$. The size of the region inside the contour is controlled by $|\Im w|$.

We set the function $f(z)$ from Section \ref{ssec:chiralstate} to be
\beq \label{eq:holeshape}
f(z) = u_{\Lambda}(e_1 z^{-2} e^{4 \pi \tau})
\eeq
so that the image of the Riemann surface after removal of $\Im w > \tau$ and $\Im w < (-\sigma)$ defines $\Sigma_{\tau,\sigma}$ where $\tau,\sigma>0$. The image of the region $(-\sigma) \leq \Im w \leq \tau$ on a single cylinder defines the elementary block (Fig. \ref{fig:Ttensor}). It is convenient to work with such $\Sigma_{\tau,\sigma}$ since the parameters $\tau$ and $\sigma$ have a simple geometric interpretation on the cylinder. In the limit $\sigma \to \infty$ we get the surface $\Sigma_{\tau}$.

\begin{figure}
    \centering
    \begin{tikzpicture} [x=30pt,y=30pt]
        \filldraw[color=teal!30, fill=teal!30] (-1,-1.5) rectangle (1,1.5);
        \filldraw[color=teal!30, fill=teal!30] (0,-1.5) ellipse (1 and 0.5);
        \filldraw[fill=teal!10] (0,1.5) ellipse (1 and 0.5);
        \draw[color=blue,very thick] (0,-1.5) + (180:1 and 0.5) arc   (180:360:1 and 0.5);
        \draw[violet, densely dashed] (-1,-1.5) -- (-1,0);
        \draw[violet, densely dashed] (1,-1.5) -- (1,0); 
        \draw[color=black!50, densely dashed] (0,0) + (180:1 and 0.5) arc   (180:360:1 and 0.5);        
        \filldraw[color = white, fill=white] (-.03,-2.1) rectangle (.03,-.5);
        \draw[color = violet, densely dashed, thick] (-.03,-2) -- (-.03,-.5) -- (.03,-.5) -- (.03,-2);
        \draw[<->] (-1.5,0) -- (-1.5,1.5);
        \draw[<->] (-1.5,-1.5) -- (-1.5,0);
        \node[] at (-2,0.5) {$\tau$};
        \node[] at (-2,-.5) {$\sigma$};
    \end{tikzpicture}\hspace{2cm}%
    \qquad 
    \begin{tikzpicture} [x=50pt,y=50pt]
        \clip (-1.3,-1.3) rectangle (1.3,1.3);
        \filldraw[color=violet, densely dashed, fill=teal!30,thick] (-1,-1) rectangle (1,1);
        \draw [rounded corners=25pt, rotate around={45:(0,0)}, fill=white]  (0-.9/1.41,0-.9/1.41) rectangle (0 + .9/1.41,0 + .9/1.41);
        \foreach \x in {-1,1}
            \foreach \y in {-1,1}
                \filldraw [rounded corners=20pt, rotate around={45:(\x,\y)},color = blue,fill=white, very thick]  (\x-.7/1.41,\y-.7/1.41) rectangle (\x + .7/1.41,\y + .7/1.41);
        \draw[color=black!50, densely dashed] (0,-1)--(1,0)--(0,1)--(-1,0)--(0,-1); 
        \foreach \a in {0,90,...,270}
            {
            \def\sd{1.47};            
            \filldraw[color=white, fill = white, rotate = \a] (\sd,0)--(\sd,\sd)--(0,\sd)--(\sd,0);
            }
        \node[] at (0,1.15) {$\Hleg$};
        \node[] at (0,-1.15) {$\Hleg$}; 
        \node[] at (1.15,0) {$\Hleg$}; 
        \node[] at (-1.15,0) {$\Hleg$};  
        \node[] at (0,0) {$\CV \otimes \overline{\CV}$};    
    \end{tikzpicture}
    \caption{The cylinder and its image under the map $u_{\Lambda}(e_1 e^{-4 \pi i w})$ to the unit square. The region $(-\sigma) \leq \Im w \leq \tau$ gives the elementary block.}
    \label{fig:Ttensor}
\end{figure}
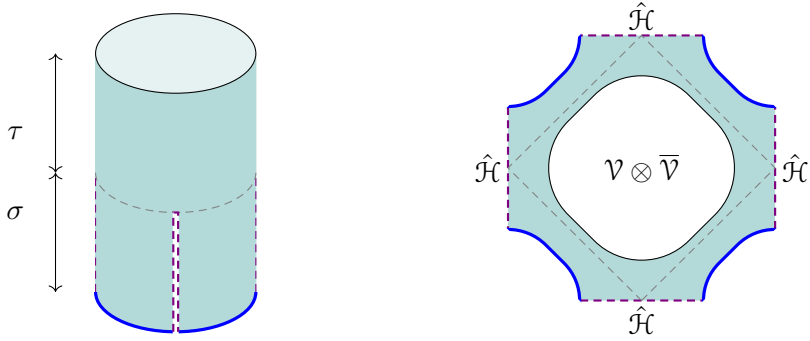

If we only remove $\Im w < (-\sigma)$, the cylinder defines the map $\lal T_{\sigma}|: \Hleg^{\otimes 4} \to \CCC$ that is the amplitude of the elementary block (with filled inner hole) with fixed states on the necks. It vanishes on basis elements $|e^{(a_1 b_1)}_{m_1} ... e^{(a_n b_n)}_{m_n} \ral$ that do not satisfy $b_k = a_{k+1}$. To express it through correlation functions of CFT on the unit disk\footnote{See Section 6 of \cite{Runkel21} for a similar computation.}, one should find a holomorphic map $w \to \xi(w)$ of the elementary block into the unit disk as shown in Fig. \ref{fig:TensorDisk} and functions $g_0(z)$,..., $g_{n-1}(z)$ such that $g_k$ maps the upper half of the unit disk into the white region near $e^{\frac{2 \pi i k}{n}}$. Then
\beq\label{eq:Ttensor1}
\lal T_{\sigma}| e^{(a_1 a_2)}_{m_1} ... e^{(a_n a_1)}_{m_n} \ral =  e^{A_L(\sigma,\tau)} (S_{{\bf 1 1}})^{3/2} \lal \biggl( \prod_{k=1}^{n} \holeO_{g_k}^{(a_k a_{k+1})}(e^{(a_k a_{k+1})}_{m_k}) \biggr) \ral^{\text{CFT}}_{D}
\eeq
where $e^{A_L}$ is the overall Liouville factor that does not depend on the boundary conditions or states and is not relevant in the following. We can take
\beq
\xi(w) = \left(\frac{1+\sqrt{\text{cosh}^2(n \pi \sigma) \tanh^2(n \pi i (w+i \sigma))-\sinh ^2(n \pi  \sigma)}}{\text{cosh}(n \pi \sigma)(1-\tanh (n \pi i (w+i \sigma)))} \right)^{1/n},
\eeq
\beq
g_k(z) = e^{\frac{2 \pi i k}{n}} \biggl( \frac{i- \tanh \l n \pi \sigma/2 \r  z}{i+ \tanh \l n \pi \sigma/2 \r  z}\biggr)^{1/n}.
\eeq
Note that when $\sigma \to 0$, the components $\lal T_{\sigma} | e^{(a_1 a_2)}_{m_1} ... e^{(a_n a_1)}_{m_n}\ral$ are suppressed by $\tanh \l n \pi \sigma/2 \r$ to the power of the total weight of basis elements, as can be seen from the rotation map eq. (\ref{eq:Rphi}). That allows us to get an upper bound on the components. More precisely, for any $\beta$ we can find ${\sigma}_0$, such that for any $0<\sigma<\sigma_0$ we have
\beq
\frac{|\lal T_{\sigma} | e^{(a_1 a_2)}_{m_1} ... e^{(a_n a_1)}_{m_n}\ral |}{|\lal  T_{\sigma} |e^{({\bf 1 1})}_{0} ... e^{({\bf 1 1})}_{0} \ral |} \leq  \exp \l - \beta \sum_{k=1}^n h^{(a_k a_{k+1})}(m_k) \r.
\eeq
The map $\lal T_{\tau}|$ defines the statistical model for the computation of averages in the state $\Psi_{\tau}$ as we discuss in Appendix \ref{app:clusterCFT}.

When $\tau$ is finite, the cylinder defines the map $\lal T_{\tau,\sigma}|: \CV \otimes \overline{\CV} \otimes \Hleg^{\otimes 4}  \to \CCC$. To express it through correlation functions of CFT on the unit disk, in addition to $g_k$, we need the function $h(z) = \xi (\frac{1}{2 \pi i}\log(z e^{-2 \pi \tau}))$ which maps the unit disk to the white region in the center in Fig. (\ref{fig:TensorDisk}). 
We have
\beq \label{eq:Ttensor}
\lal T_{\tau,\sigma} |\alpha \bar{\alpha} e^{(a_1 a_2)}_{m_1}... e^{(a_n a_1)}_{m_n}\ral  = e^{A_L(\sigma,\tau)} (S_{{\bf 1 1}})^{3/2} \lal \biggl( \prod_{k=1}^{n} \holeO_{g_k}^{(a_k a_{k+1})}(e^{(a_k a_{k+1})}_{m_k}) \biggr) \holeO_{h}(\alpha \otimes \bar{\alpha}) \ral^{\text{CFT}}_{D}
\eeq
where $|\alpha \bar{\alpha}\ral \in \CV \otimes \overline{\CV}$. The tensor network state defined using the map $\lal T_{\tau,\sigma}|$ gives the state $\Phi_{\tau}(\sigma)$.

\begin{figure}
    \centering
    \begin{tikzpicture}[x=80pt,y=80pt]
        \centering
        \filldraw[color=white, fill=teal!30] (0,0) circle (1);
        \filldraw[fill=white!30] (0,0) circle (0.25);
        \foreach \a in {0,90,...,270}
            {
            \filldraw[color = violet, densely dashed, fill=white, rotate=\a] (1,0) + (90:0.5 and 0.125) arc   (90:270:0.5 and 0.125);
            }
        \draw[color=blue, very thick] (0,0) circle (1);
        \foreach \a in {0,90,...,270}
            {
                \filldraw[white, rotate = \a] (1,0) circle (0.12);    
                \draw[color=blue, densely dashed, rotate = \a] (0,0) + (-10:1 and 1) arc (-10:10:1 and 1);
                \filldraw[black, rotate = \a] (1,0) circle (2pt);
            }
        \filldraw[black] (0,0) circle (2pt);
    \end{tikzpicture}
    \caption{The images of the maps $g_0$, $g_1$, $g_2$, $g_3$, $h$ are shown in white with the black dots being the images of 0.}
    \label{fig:TensorDisk}
\end{figure}
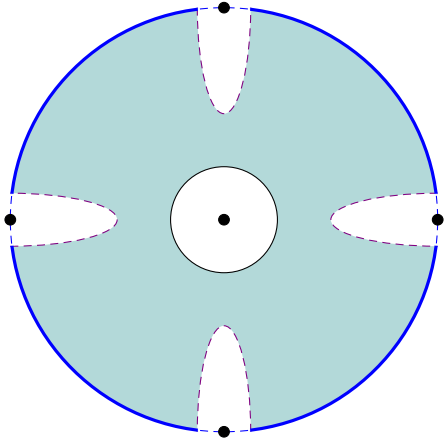

\begin{remark}
We have discussed only the maps $ \Hleg^{\otimes 4} \to \CCC$ and $\CV \otimes \overline{\CV} \otimes \Hleg^{\otimes 4} \to \CCC$ corresponding to the blocks in the interior of the network. The maps corresponding to blocks on the boundary can be defined by contracting the outward legs with some elements of $\Hleg$, corresponding to attaching the boundaries, in a straightforward way.
\end{remark}

\section{Cluster expansion} \label{app:cluster}

In this section, we first recall standard facts about the cluster expansion for statistical models, also known as the polymer expansion. We follow Ch. 5 of \cite{friedli2017statistical}. We then apply it to tensor networks coming from our construction. We will use it to prove that the states $\Psi_{\tau}$ defined in the main text are well-defined in the thermodynamic limit, at least for a small enough value of $\tau$, and that the correlators decay exponentially with the distance.

\subsection{General setup}

Let $\Gamma$ be a set equipped with a function $w:\Gamma \to \CCC$ and a binary relation that contains the diagonal. We call elements of $\Gamma$ polymers; $w(\gamma)$ is the weight of a polymer $\gamma \in \Gamma$. If $(\gamma_1,\gamma_2)$ does not belong to the binary relation, we write $\gamma_1 \cap \gamma_2 = \emptyset$ and say that the polymers do not intersect. Otherwise, we write $\gamma_1 \cap \gamma_2 \neq \emptyset$. A subset of polymers $\Gamma' \subset \Gamma$ is called disconnected if all its elements do not intersect with each other. The partition function is defined by
\beq
Z = \sum_{\Gamma'} \prod_{\gamma \in \Gamma'} w(\gamma)
\eeq
where the sum is over all disconnected subsets $\Gamma' \subset \Gamma$.

We call a collection $X$ of elements of $\Gamma$ a cluster if the graph of intersections\footnote{The graph of intersections is the graph with vertices being the elements of $X$ and edges being pairs $\{\gamma_1,\gamma_2\}$ such that $\gamma_1 \cap \gamma_2 \neq \emptyset$.} $G(X)$ of $X$ is connected. The index of a cluster $X$ is defined to be
\beq
I(X) : = \sum_{H} (-1)^{|E(H)|}
\eeq
where the sum is over all connected subgraphs $H$ of $G(X)$ and $|E(H)|$ is the number of edges in $H$. If a collection $X$ is not a cluster, we set $I(X) = 0$. Then, formally, we have
\begin{multline} \label{eq:CE}
\log Z = \sum_{n \geq 1} \sum_{\gamma_1 \in \Gamma} ... \sum_{\gamma_n \in \Gamma} \frac{1}{n!} I(\{\gamma_1,...,\gamma_n\}) \prod_{i=1}^n w(\gamma_i) = \\ = \sum_{X \in C} \frac{I(X)}{\prod_{\gamma \in \Gamma} n_X(\gamma)!} \prod_{\gamma \in X} w(\gamma) =: \sum_{X \in C} \clusterweight(X)
\end{multline}
where $n_X(\gamma)$ is the number of times the element $\gamma$ appears in $X$ and $C$ is the set of clusters. This expansion is convergent provided a certain sufficient condition is satisfied. We will use the following criterion (Theorem 5.4 from \cite{friedli2017statistical}). Suppose there is a function $|\cdot|: \Gamma \to \RR_{> 0}$ such that for any $\gamma \in \Gamma$ we have
\beq \label{eq:CEconvcrit}
\frac{1}{|\gamma|} \sum_{\gamma':\gamma \cap \gamma' \neq \emptyset} |w(\gamma')| e^{|\gamma'|} \leq 1
\eeq
and 
\beq \label{eq:CEconvcrit2}
\sum_{\gamma'} |w(\gamma')| e^{|\gamma'|} < \infty.
\eeq
Then for any $\gamma_1 \in \Gamma$ we have
\beq \label{eq:CEconvcrit3}
1+ \sum_{n\geq 2} \sum_{\gamma_2}...\sum_{\gamma_n} \frac{|I(\{\gamma_1,...,\gamma_n\})|}{(n-1)!} \prod_{i=2}^{n} |w(\gamma_i)| \leq e^{|\gamma_1|}.
\eeq
In particular, it guarantees the convergence of the cluster expansion eq. (\ref{eq:CE}).

\subsection{CFT tensor network} \label{app:clusterCFT}

\begin{figure}
    \centering
    \begin{tikzpicture} [x=15pt,y=15pt]
        \draw [blue, fill=teal, fill opacity = 0.2] (-5,-5) rectangle (5,5);
        \draw [violet, dashed] (-6,-6) -- (-6,6);
        \draw [violet, dashed] (-6,6) -- (6,6);
        \foreach \x in {-4,-2,...,4}
            \foreach \y in {-4,-2,...,4}
                {
                \draw [blue, fill=white, rounded corners=6pt, rotate around={45:(\x,\y)}]  (\x-.9/1.41,\y-.9/1.41) rectangle (\x +.9/1.41,\y +.9/1.41);
                }
        \foreach \x in {-4,-2,...,6}
            \foreach \y in {-6,-4,...,4}
                {
                \draw [violet, dashed] (\x,\y) -- (\x-2,\y);
                \draw [violet, dashed] (\x,\y) -- (\x,\y+2);        
                }
      
   
        \foreach \x in {-5,-3,...,5}
            \foreach \y in {-5,-3,...,5}
                {
                \filldraw[red] (\x,\y) circle (2pt);
                }
    \end{tikzpicture}
    \caption{Each red point is the vertex of $V_G$. The edges $E_G$ are not shown, but they connect two neighboring red points. A square containing some vertex $v$ corresponds to the piece of the surface that defines the vector $|T^{(v)}\ral \in \otimes_{e \ni v} \Hleg $.}
    \label{fig:cut1}
\end{figure}
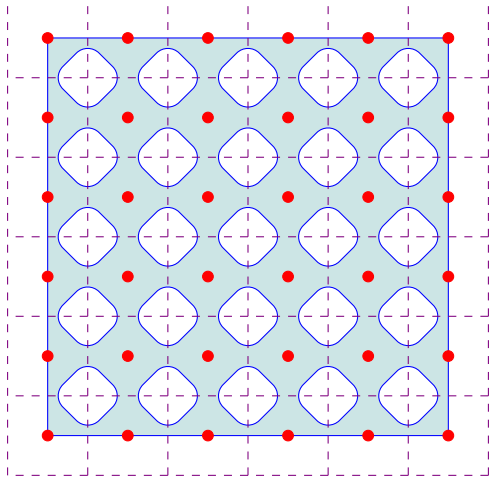

To compute the averages of observables in the state $\Psi_{\tau,\Gamma_*}$, we introduce an auxiliary statistical model. 

Let $V$, $E$, $F$ be the set of vertices, edges, and faces of the infinite square grid defined by the dual lattice $\dLat$. We choose the same orientation on each edge and face. Let $G$ be the subgraph of size $(2N) \times (2N)$ as show in Fig. (\ref{fig:cut1}). The corresponding sets of vertices, edges, and faces are denoted by $V_G$, $E_G$, $F_G$, respectively. Let's cut the surface $\Sigma_{\tau}$ into simple blocks, as shown in Fig. \ref{fig:cut1}. Each block containing a vertex $v \in V_G$ defines a map $\lal T^{(v)}| : \bigotimes_{e \ni v }\Hleg \to \CCC$ where $e \in E_G$. The components of the map can be computed as explained in Appendix \ref{app:basictensor}. The space of states of the model is $\CH_G = \bigotimes_{e \in E_G} \Hleg_e$, where $\Hleg_e \cong \Hleg$ is the space associated with each edge of $E_G$. Using a natural pairing on each $\Hleg_e$, the contracted maps $\lal T^{(v)}|$ compute the partition function of the CFT on $\Sigma_{\tau}$ with the cloaking combination of elementary boundary conditions. We label the basis elements of $\CH_{G}$ by a function $P:F \to \ObRepV$ that is constant on $f \not \in F_G$, and a function $J:E \to \NN_0$ with $J(e)=0$ for $e \not \in E_G$. A pair $(P,J)$ is called a configuration. The function $P$ defines the choice of elementary boundary conditions on each boundary, while the function $J$ defines the corresponding vectors in $\Hleg$. Hence, each configuration $(P,J)$ defines a basis element $|v,P,J \ral \in \otimes_{e \ni v} \Hleg_e$ for each vertex $v \in V_G$. We define the weight $W(P,J)$ of a configuration $(P,J)$ by the contraction of $\lal T^{(v)}|v,P,J\ral \lal v,P,J|$. The partition function of the model is given by
\beq
Z = \sum_{P,J} W(P,J).
\eeq

Let $A$ be a connected set of faces on the lattice $\Lambda$, and let $S = \{s_1,...,s_k\}$ be the set of edges of $E_G$ intersecting the boundary of $A$ (see Fig. \ref{fig:cut2}). The contraction of maps $\lal T^{(v)}|$ for $v$ that belong to faces of $A$ defines a map $\lal T_A| : \bigotimes_{s \in S} \Hleg_s \to \CCC$ or the corresponding vector $|T_A \ral \in \bigotimes_{s \in S} \Hleg_s$. Similarly, we can define $\lal T_{\bar{A}}|$ for the complementary set $\bar{A}:=F_G \backslash A$. The partition function is given by $\lal T_{\bar{A}}| T_{A}\ral$

Any local observable $\CA \in \SA$ with the support inside $A$ gives another canonical map $\lal O(\CA)| : \bigotimes_{s \in S} \Hleg_s \to \CCC$. Note that $\lal O(1)| = \lal T_A|$. The average of the observable $\CA$ can be expressed by
\beq \label{eq:averageTA}
\lal \CA \ral_{\Psi_{\tau,\Gamma_*}} = \frac{\lal T_{\bar A} | O(\CA) \ral}{\lal T_{\bar A} | T_A \ral}.
\eeq
Note that $\||O(\CA)\ral\| \leq \|\CA\| \||T_A \ral \|$.

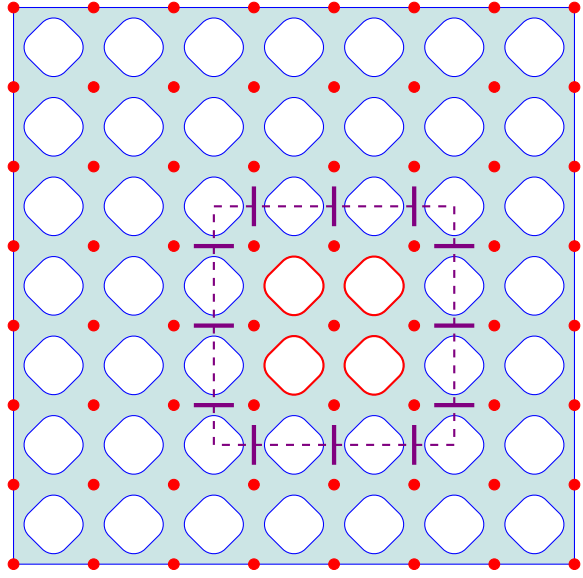
\begin{figure}
    \centering
    \begin{tikzpicture} [x=15pt,y=15pt]
        \draw [blue, fill=teal, fill opacity = 0.2] (-7,-7) rectangle (7,7);
        \foreach \x in {-6,-4,...,6}
            \foreach \y in {-6,-4,...,6}
                {
                \draw [blue, fill=white, rounded corners=6pt, rotate around={45:(\x,\y)}]  (\x-.9/1.41,\y-.9/1.41) rectangle (\x +.9/1.41,\y +.9/1.41);
                }
        \draw [violet, dashed, thick] (-2,-4) -- (-2,2) -- (4,2) -- (4,-4) -- (-2,-4);
        \foreach \x\y in {-2/1,-2/-1,-2/-3,4/1,4/-1,4/-3}
            {
            \draw[violet,ultra thick] (\x-0.5,\y) -- (\x+0.5,\y);
            }
        \foreach \x\y in {-1/2,1/2,3/2,3/-4,1/-4,-1/-4}
            {
            \draw[violet,ultra thick] (\x,\y-0.5) -- (\x,\y+0.5);
            }            
                
        \foreach \x\y in {0/0,2/0,0/-2,2/-2}
                {
                \draw [red, thick, fill=white, rounded corners=6pt, rotate around={45:(\x,\y)}]  (\x-.9/1.41,\y-.9/1.41) rectangle (\x +.9/1.41,\y +.9/1.41);                    
                }
        \foreach \x in {-7,-5,...,7}
            \foreach \y in {-7,-5,...,7}
                {
                \filldraw[red] (\x,\y) circle (2pt);
                }
    \end{tikzpicture}
    \caption{$A$ consists of four faces forming a square. The edges of $S$ are parallel to violet segments.}
    \label{fig:cut2}
\end{figure}

In the following, we analyze the model in the regime of small $\tau$. When $\tau$ is small, the system energetically prefers the configurations with constant $P$ and $J(e)=0$, and we can perform the cluster expansion around this point. We show that there is $\tau_0$ such that for any $0 < \tau < \tau_0$ the cluster expansion converges, and the correlations decay exponentially. We use the fact that for any $\beta$ there is $\tau_{\beta}$ such that for $\tau < \tau_{\beta}$ the following bound holds
\beq
\frac{|\lal T^{(v)}| e^{(a_1 a_2)}_{m_1} ... e^{(a_n a_1)}_{m_n} \ral|}{|\lal T^{(v)}| e^{(aa)}_{0} ... e^{(aa)}_{0} \ral |} \leq  \exp \l - \beta \sum_{k=1}^n h^{(a_k a_{k+1})}(m_k) \r
\eeq
as argued in Appendix \ref{app:basictensor}.

We first analyze the model for a holomorphic $\voa$, so that $I = \{ {\bf 1} \}$. Then we discuss the modifications for general $\voa$.

\subsubsection{Single vacuum}

Suppose $\voa$ is holomorphic. In this case, there is only one sector $a = {\bf 1}$, and the configurations are labeled only by functions $J$. We omit $P$ in this subsection. We also renormalize all $\lal T^{(v)}|$ so that the lowest weight components are equal to $1$. For a configuration $J$, we say that an edge $e$ is activated if $J(e) > 0$. We can apply the cluster expansion with polymers being connected sets of edges and with the weight of a polymer being the sum of $W(J)$ over configurations $J$ such that the set of activated edges coincides with the polymer.

Let $|\gamma|$ be the number of edges in $\gamma$. Then for any $e \in E_G$ we have
\beq
\sum_{\gamma \ni e} |w(\gamma)| e^{|\gamma|} \leq \sum_{l \geq 1} n_l e^{l} \l \sum_{m=1}^{\infty} e^{-2\beta h(m)} \r^{l}
\eeq
where $n_l$ is the number of connected subsets of $E_G$ intersecting $e$ of cardinality $l$. The sum in the brackets is related to the character of $\voa$, and for rational $\voa$ with modular $\text{Rep}(\voa)$ can be made arbitrarily small by varying $\beta$. Since\footnote{That can be argued as follows. For each connected subset of edges and a vertex $v$ that belongs to it, there is a path of length $2l$ that starts at $v$ and crosses each edge exactly twice. The number of all paths of length $2l$ starting at some fixed vertex of $e$ and crossing $e$ is bounded by $4^{2l}$.} $n_l \leq 4^{2l}$, for a large enough $\beta$ the cluster expansion convergence criteria eq. (\ref{eq:CEconvcrit}) and eq. (\ref{eq:CEconvcrit2}) are satisfied for any $N$. Moreover, by eq.(\ref{eq:CEconvcrit3}), for large enough $\beta$ we have
\beq
\sum_{X:\text{supp}(X) \ni e} |\CW(X)| e^{c |X|} \leq 1
\eeq
and therefore
\beq \label{eq:ClusterExpBound}
\sum_{\substack{X:  |X| \geq R \\ \text{supp}(X) \ni e}} |\CW(X)| \leq e^{- c R}
\eeq
for any $e \in E_G$ and $R > 0$, where $\CW(X)$ is the weight of the cluster defined in eq. (\ref{eq:CE}), $|X|$ is the number of edges appearing in the polymers of $X$ and $c$ is some constant.

The averages eq. (\ref{eq:averageTA}) correspond to the contraction with $(\lal T_{\bar A} | T_A \ral)^{-1} | T_{\bar A} \ral$. Let us write $| T^{(N)}_{\bar A} \ral$ in this paragraph to indicate the dependence on $N$. We can use the cluster expansion to compute the overlap of normalized vectors for consecutive $N$
\beq
\frac{\lal T^{(N+1)}_{\bar A} | T^{(N)}_{\bar A} \ral \lal T^{(N)}_{\bar A} | T^{(N+1)}_{\bar A} \ral}{\lal T^{(N+1)}_{\bar A} | T^{(N+1)}_{\bar A} \ral \lal T^{(N)}_{\bar A} | T^{(N)}_{\bar A} \ral}
\eeq
with the polymers and clusters living in the tensor networks obtained by gluing two complements of $A$ of the original networks along the edges of $S$. Note that the clusters which do not intersect simultaneously the boundary and $S$ cancel. The remaining clusters have at least as many edges as the distance between $A$ and the boundary, and eq.(\ref{eq:ClusterExpBound}) implies that their contribution decays exponentially with $N$. Similarly, the norms of $(\lal T^{(N)}_{\bar A} | T_A \ral)^{-1} \lal T^{(N)}_{\bar A} |$ are upper bounded, and since they have unit overlap with $|T_{A}\ral$, the thermodynamic limit $N \to \infty$ exists. Therefore the state $\Psi_{\tau}$ is well-defined.

Suppose now we have two observables $\CA \in \BoundedOps(\CV_{A})$ and $\CB \in \BoundedOps(\CV_{B})$ localized in disjoint collections of plaquettes $A$ and $B$ with the corresponding sets of edges $S_A$ and $S_B$ intersecting the boundaries of the regions. To find a bound on $|\lal \CA \CB \ral_{\Psi_{\tau}} - \lal \CA \ral_{\Psi_{\tau}} \lal \CB \ral_{\Psi_{\tau}}|$, we can compute the normalized overlap of vectors $|T_{\overline{AB}} \ral$ and $|T_{\bar{A}} T_{\bar{B}} \ral := |T_{\bar{A}}\ral \otimes |T_{\bar{B}} \ral$ corresponding to the contraction of $|T^{(v)}\ral$ in the complement of the corresponding regions. The cluster expansion of
\beq
\frac{\lal T_{\overline{AB}}| T_{\bar{A}} T_{\bar{B}} \ral \lal T_{\bar{A}} T_{\bar{B}} | T_{\overline{AB}} \ral}{\lal T_{\bar{A}} T_{\bar{B}}|T_{\bar{A}} T_{\bar{B}} \ral \lal T_{\overline{AB}}|T_{\overline{AB}}\ral}
\eeq
on the corresponding glued tensor networks has the contribution only from clusters intersecting both $A$ and $B$. They have at least as many edges as the shortest distance $R$ between $A$ and $B$. Therefore using eq.(\ref{eq:ClusterExpBound}) one can get
\beq
\|\rho_{AB}-\rho_{A} \otimes \rho_B\|_1 \leq C \text{min}(|S_A|,|S_B|) e^{-\alpha R}
\eeq
for some constants $C,\alpha$ and for density matrices $\rho_A$, $\rho_B$, $\rho_{AB}$ corresponding to the restrictions of the state $\Psi_{\tau}$ to the supports $A$ and $B$ of the observables $\CA$ and $\CB$, respectively. In particular, we have
\beq
|\lal  \CA \CB \ral_{\Psi_{\tau}}  - \lal \CA \ral_{\Psi_{\tau}} \lal \CB \ral_{\Psi_{\tau}}| \leq C \, \text{min}(|S_A|,|S_B|) \|\CA \| \|\CB\| e^{-\alpha R}.
\eeq Thus correlations between observable in the state $\Psi_{\tau}$ have exponential decay. Similarly, if one modifies the state $\Psi_{\tau}$ by a local insertion of vertex operators (as in Section \ref{ssec:perturb}), the change of the expectation values of observables decays exponentially with the distance between the insertions and the support of the observable.

\subsubsection{Multiple vacua}

For a non-holomorphic $\voa$, the polymers would be again the subsets of connected edges of $E_G$. For a configuration $(P,J)$, an edge $e$ is activated (i.e. the state running through the corresponding neck has a non-zero weight) if $J(e)>0$ or if $P$ has different values on the adjacent faces.  The weight of a polymer is given by the sum of $W(P,J)$ over all configurations $(P,J)$ which set of activated edges coincides with the polymer. Note that $P$ of the configurations contributing to the weight has the same value in each connected component of the partition of the plane defined by the polymer. For a given $P$, we can replace the elementary boundary condition around each hole (defined by $P$) with the trivial elementary boundary condition and the appropriate topological defect around the hole. For the computation of the weight of a given polymer $\gamma$, we can reconnect all the topological defects in deactivated necks as
the states running through those necks must have zero weight. Each reconnection of defects with label $a$ gives a factor $(S_{{\bf 1}a}/S_{{\bf 1}{\bf 1}})^{-1}$. The resulting configuration has contractible topological defects near the vertices inside each connected component (which give a factor $(S_{{\bf 1}a}/S_{{\bf 1}{\bf 1}})$) and topological defects running along the edges of the polymer. That allows to show that the contribution of a collection of non-intersecting polymers to the partition function is given by the product of weights. We have a bound 
\beq
|w(\gamma)| \leq \l \sum_{a,b} \sum_{m:h^{(ab)}(m)>0} e^{- \beta h^{(ab)}(m)} \r^{|\gamma|}
\eeq
that can be made arbitrarily small by varying $\beta$, and therefore, in the same way as for holomorphic $\voa$, for a large enough $\beta$ the cluster expansion convergence criteria eq. (\ref{eq:CEconvcrit}) and eq. (\ref{eq:CEconvcrit2}) are satisfied for any $N$.

We can organize the cluster expansion in a similar way to show that, at least for large enough $\beta$ (or small enough $\tau$), the averages of observables $\lal\CA\ral_{\Psi_{\tau}}$ are well-defined in the thermodynamic limit, and the correlators of two local observables decay exponentially with the distance between their supports. Modifications of the state $\Psi_{\tau}$ by a local insertion of vertex operators can also be treated similarly, while for insertions of non-trivial modules (see Section \ref{ssec:anyons}) the arguments are no longer applicable since it may not be possible to contract the topological defects after the reconnections. That is consistent with the fact that the averages of observables on an annulus might be affected by an insertion of an anyon into the center even if the annulus is very large.

\section{Local perturbations} \label{app:localpert}

Suppose we have two pure states $\psi$, $\psi'$ of the quasi-local algebra $\SA$ that coincide at infinity, i.e. for any $\eps$ there is $R$, such that for any observable $\CO \in \SA$ supported outside the disk of radius $R$ with the center at 0 we have
\beq \label{eq:atinfinity}
|\lal \CO \ral_{\psi'}-\lal \CO \ral_{\psi}| \leq \eps \|\CO\|.
\eeq
By Corollary 2.6.11 \cite{bratteli2012operator}, both states can be described as vector states in the same Hilbert space (the GNS Hilbert space associated with $\psi$). Let us choose the corresponding state vectors $|\psi\ral$ and $|\psi'\ral$. By Kadison transitivity theorem, one vector can be produced from another by a unitary element $\CU$ of $\SA$. The theorem also implies that there is a self-adjoint observable $\CK \in \SA$ such that $|\psi'\ral-\CK|\psi\ral$ is proportional to $|\psi\ral$.

We say that a state $\psi$ of a lattice system has an exponential decay of mutual correlations if for any finite disjoint subsets $X,Y \subset \Lambda$ and any $\CO \in \SA_{X \cup Y}$, we have
\beq
|\lal \CO \ral_{\psi} - \lal \CO \ral_{\psi|_X \otimes \psi|_Y}| \leq C e^{-a r} \|\CO\|,
\eeq
where $r$ is the distance between $X$ and $Y$, and $C,a>0$ are some constants. Equivalently, the reduced density matrices $\rho_X, \rho_Y, \rho_{XY}$ for subsets $X$,$Y$, $X \cup Y$, respectively, satisfy
\beq
\|\rho_{XY} - \rho_X \otimes \rho_Y\|_1 \leq C e^{-a r}.
\eeq

Suppose two pure states $\psi$ and $\psi'$ have exponential decay of mutual correlations and $\eps$ in eq. (\ref{eq:atinfinity}) decays exponentially with the corresponding $R$. In this appendix, we argue that for such states both $\CU$ and $\CK$ can be chosen to be almost local, i.e. they can be approximated by local observables with the support on disks of radius $r$ with the error that decays rapidly (i.e. faster than any power) with $r$. As argued in Appendix \ref{app:clusterCFT}, at least for sufficiently small $\tau$, the state $\Psi_{\tau}$ and the state modified by a local insertion of vertex operators into the defining correlator eq. (\ref{eq:chiralstate2}) satisfy these properties.

We will use the following

\begin{lemma} \label{lma:Uapprox}
Let $\eps,\delta \in [0,1]$. Let $|\chi\ral$ and $|\chi'\ral$ be unit vectors on $\CH_A \otimes \CH_B \otimes \CH_C$, such that 1) $\||\chi'\ral - |\chi\ral\| \leq \sqrt{\eps}$, 2) $\|\rho'_{BC}-\rho_{BC}\|_{1} \leq \delta$, 3) $\|\rho_{AC} - \rho_{A} \otimes \rho_{C}\|_1 \leq \delta^2$ and $\|\rho'_{AC} - \rho'_{A} \otimes \rho'_{C}\|_1 \leq \delta^2$, where $\rho_X$ and $\rho'_X$ are density matrices for the restrictions of $|\chi\ral \lal \chi |$ and $|\chi'\ral \lal \chi'|$ to $X$, respectively. Then there is a unitary $\CU_{AB} \in \CB(\CH_A \otimes \CH_{B})$ such that $\|\CU_{AB}-1\| \leq 3 \sqrt{\delta} + \sqrt{\eps}$ and $\||\chi'\ral - \CU_{AB}|\chi\ral\| \leq 3 \sqrt{\delta}$.
\end{lemma}
\begin{proof}
We will repeatedly use the Fuchs van de Graaf inequalities for the fidelity $F(\rho,\sigma) = (\| \sqrt{\rho} \sqrt{\sigma}\|_1)^2$ of quantum states $\rho$, $\sigma$
\beq
1-\sqrt{F(\rho,\sigma)} \leq \frac12 \|\rho - \sigma\|_1 \leq \sqrt{1-F(\rho,\sigma)}
\eeq
and Uhlmann's theorem that identified the fidelity with the maximal overlap of purifications $|\lal \chi_{\rho}|\chi_{\sigma} \ral|^2$. We also use that $\||\phi'\ral - |\phi\ral\| \leq \sqrt{\epsilon}$ is equivalent to $\Re \lal \phi'|\phi\ral \geq 1- \frac{\epsilon}{2}$ for any unit vectors $|\phi'\ral$, $|\phi\ral$.

By the first Fuchs van de Graaf inequality and Uhlmann's theorem for $\rho_{AC}$ and $\rho_A \otimes \rho_C$, there is a purification $|\tilde{\chi} \ral$ of $\rho_A \otimes \rho_C$ such that $|\lal\tilde{\chi}|\chi \ral| \geq (1-\delta^2/2)$ and $\| |\tilde{\chi}\ral-|\chi \ral \| \leq \delta$. Let $\tilde{\rho}_X$ be the density matrix for the restriction of $|\tilde{\chi}\ral \lal \tilde{\chi}|$ to $X$. By the second Fuchs van de Graaf inequality we have $\|\tilde{\rho}_{BC} - \rho_{BC}\|_1 \leq 2 \delta$ and therefore $\|\rho'_{BC} - \tilde{\rho}_{BC}\|_1 \leq 3 \delta$. By the first Fuchs van de Graaf inequality and Uhlmann's theorem, there is a 
vector $|\tilde{\chi}' \ral$ that purifies $\tilde{\rho}_{BC}$ such that $\lal \chi'|\tilde{\chi}'\ral \geq 1 - 3 \delta/2$. Since $\|\tilde{\chi}'\ral - |\tilde{\chi}\ral\| \leq \sqrt{3 \delta} + \sqrt{\eps} + \delta \leq 3 \sqrt{\delta} + \sqrt{\eps}$ and since the states $|\tilde{\chi}\ral \lal \tilde{\chi}|$ and $|\tilde{\chi}'\ral \lal \tilde{\chi}'|$ have no correlations between $A$ and $C$ and coincide on $BC$ there is a unitary $\CU_{AB} \in \CB(\CH_A \otimes \CH_{B})$ such that $\|\CU_{AB}-1\| \leq 3 \sqrt{\delta} + \sqrt{\eps}$ and $\CU_{AB}|\tilde{\chi}\ral=|\tilde{\chi}'\ral$. We have $\|\CU_{AB}|\chi\ral - |\chi'\ral\| \leq \delta + \sqrt{3 \delta} \leq 3 \sqrt{\delta}$.
\end{proof}

\begin{lemma}  \label{lma:Kapprox}
Let $|\chi\ral$ be a unit vector on $\CH_{A} \otimes \CH_{B} \otimes \CH_{C}$ and $\CA \in \CB(\CH_{A})$ such that 1) $\|\rho_{AC} - \rho_A \otimes \rho_{C} \| \leq \delta^2$, 2) $\|\CA\| = 1$ and $\|\CA|\chi\ral\| \leq \eps$. Then there is a constant $c \in \CCC$ and a self-adjoint $\CK \in \CB(\CH_{BC})$ such that $\|c + \CK\| \leq  2(\eps+\delta)$ and $\|(c+\CK-\CA)|\chi\ral \| \leq \delta (1+2\eps+2\delta)$.
\end{lemma}
\begin{proof}
By the first Fuchs van de Graaf inequality and Uhlmann's theorem for $\rho_{AC}$ and $\rho_A \otimes \rho_C$, there is a purification $|\tilde{\chi} \ral$ of $\rho_A \otimes \rho_C$ such that $|\lal\tilde{\chi}|\chi \ral| \geq (1-\delta^2/2)$ and $\| |\tilde{\chi}\ral-|\chi \ral \| \leq \delta$. Then $\|\CA|\tilde{\chi}\ral\| \leq \eps + \delta$. Since $|\tilde{\chi}\ral \lal \tilde{\chi}|$ has no correlations between $A$ and $C$, there is $c \in \CCC$ and a self-adjoint $\CK \in \CB(\CH_{A} \otimes \CH_{B})$ with $|c|,\|\CK\|\leq (\eps + \delta)$ such that $\CA|\tilde{\chi}\ral = (c + \CK) |\tilde{\chi}\ral$. We have $\|c+\CK\| \leq 2(\eps+\delta)$ and $\|(c+\CK-\CA)|\chi\ral\| \leq \|(c+\CK-\CA)\| \||\tilde{\chi}\ral - |\chi\ral \| \leq \delta (1+2\eps+2\delta)$.
\end{proof}

\begin{prop}
Let $D_r$ be the disk of radius $r$ with the center at $0$ and let $C \geq 1$, $\alpha >0$. Suppose we have vectors $|\psi\ral, |\psi'\ral \in \CH$, representing unitarily equivalent pure states $\psi$ and $\psi'$, such that (1) for any $R,r>0$ satisfying $R>r$ and any $\CO \in \SAl$ with the support outside the ring $D_R \backslash D_r$, we have
\beq
|\lal \CO \ral_{\omega} - \lal \CO \ral_{\omega|_{D_r} \otimes \omega|_{D^c_R}}| \leq C e^{- \alpha (R-r)} \|\CO\|
\eeq
for both $\omega = \psi$ and $\omega = \psi'$; (2) for any $r>0$ and any $\CO \in \SAl$ with the support outside $D_r$, we have
\beq
|\lal \CO \ral_{\psi} - \lal \CO \ral_{\psi'}| \leq C e^{- \alpha r} \|\CO\|.
\eeq
Then 1) there is an almost local unitary $\CU \in \SA$ with localization depending on $C$ and $\alpha$ only such that $|\psi'\ral = \CU |\psi\ral$; 2) there is an almost local self-adjoint observable $\CK \in \SA$ with localization depending on $C$ and $\alpha$ only such that $|\psi'\ral - \CK |\psi\ral$ is proportional to $|\psi\ral$.
\end{prop}

\begin{hproof}
By the first Fuchs van de Graaf inequality and Uhlmann's theorem, for a given $\sqrt{\eps}$ we can always find $r$ (that depends on $C$, $\alpha$ and $\eps$ only) such that there is a local unitary $\CU^{(0)}$ with the support on $D_r$ satisfying $\||\psi'\ral - \CU^{(0)}|\psi\ral \| \leq \sqrt{\eps}$.

Let $A$ be $D_r$, $B$ be the annulus between $D_r$ and the complement of $D_{3 r}$ and $C$ be the complement of $D_{3 r}$. Suppose we found a unitary $\CU_r$ with the support on $A$ such that for $|\chi\ral = \CU_r |\psi\ral$ we have $\||\psi'\ral - |\chi\ral\| \leq \sqrt{\eps}$ for some $\eps < 1$. Let $\delta = C e^{-\alpha r}$. Note that vectors $|\chi\ral$ and $|\psi'\ral$ satisfy the conditions of Lemma \ref{lma:Uapprox}, that allows us to find a unitary $\CU_{3 r}$ on the disk of radius $3r$ such that $\|\CU_{3 r} - 1\| \leq 3\sqrt{\delta} + \sqrt{\eps}$ and $\||\psi'\ral - \CU_{3 r} |\chi\ral \| \leq 3 \sqrt{\delta}$.

We can take $\CU^{(0)} = \CU_r$, $\CU^{(1)} = \CU_{3 r}$ and iterate the procedure to find unitaries $\CU^{(n)}$ localized on disks $D_{R_n}$ of radius $R_n = 3^n r$ such that $\|\CU^{(n)}-1\|$ and $\||\psi'\ral - \CU^{(n)}...\CU^{(1)}\CU^{(0)} |\psi\ral \|$ decay exponentially with $R_n$. The infinite product of all such unitaries exists and gives an almost local observable $\CU$ that produces $|\psi'\ral$ from $|\psi\ral$.

Since $|\psi'\ral = \CU |\psi\ral$, and since we can represent $\CU$ as a sum of local observables with norms rapidly decaying with the size of the support, to show the existence of $\CK$, it is enough to consider the case $|\psi'\ral = \CA |\psi\ral$ for a local observable $\CA$.

Assume the support of $\CA$ is inside $D_{r}$. Suppose we have found $c_r \in \CCC$ and a self-adjoint $\CK_r$ with the support on $D_r$ such that $\|(\CA-c_r-\CK_r)|\psi\ral \| \leq \eps \|(\CA-c_r-\CK_r)\|$ for some $\eps$. Let $\delta = C e^{- \alpha r}$. Lemma \ref{lma:Kapprox} allows us to find $c_{3r} \in \CCC$ and a self-adjoint $\CK_{3r}$ with the support on $D_{3r}$ such that $\|c_{3r}+\CK_{3r}\| \leq 2(\eps+\delta)\|(\CA-c_r-\CK_r)\|$ and $\|(\CA-c_r-\CK_r-c_{3r}-\CK_{3r})|\psi\ral \| \leq \delta(1+2 \eps + 2 \delta) \|(\CA-c_r-\CK_r)\|$.

We can iterate the procedure, in the same way as for unitaries above, to find $c^{(n)}$ and $\CK^{(n)}$ on $D_{R_n}$ such that $\|c^{(n)}+\CK^{(n)}\|$ and $\|(\CA - \sum_{l=0}^{n} (c^{(l)}+\CK^{(l)}))|\psi\ral\|$ decay exponentially with $R_n$. The sum of $\CK^{(n)}$ gives a desired $\CK$.
\end{hproof}


\printbibliography

\end{document}